\documentclass[letterpaper,journal]{IEEEtran}
\usepackage{graphicx}
\usepackage{amssymb}
\usepackage{setspace} 

\usepackage{booktabs}
\usepackage{booktabs,makecell,multirow}

\usepackage{booktabs}
\usepackage{hyperref}

\usepackage{widetext}
\usepackage{bm}
\usepackage{color, xcolor}
\usepackage{transparent}
\usepackage{import}
\usepackage[ruled]{algorithm2e}
\usepackage{amsmath}
\usepackage{amsthm}
\allowdisplaybreaks[4]
\usepackage{algorithmic}
\makeatletter

\newcommand{\Rmnum}[1]{\expandafter\@slowromancap\romannumeral #1@}
\makeatother
\usepackage{subfigure} 
\usepackage{threeparttable}

\usepackage{color}
\usepackage{transparent}
\usepackage{graphicx}
\usepackage{import}

\newtheorem{proposition}{Proposition}
\newtheorem{theorem}{Theorem}

\ifCLASSINFOpdf
\else
\fi
 \usepackage{graphicx}
\usepackage{CJKutf8}
\usepackage{bm}
\usepackage{cite}
\usepackage{amssymb}
\usepackage{amsmath}
 \usepackage{threeparttable}
\begin{document}
	\begin{CJK}{UTF8}{gbsn}
%
\title{Pipelining Split Learning in Multi-hop Edge Networks}
\markboth{IEEE Transactions on  Mobile Computing}%
{Shell \MakeLowercase{\textit{et al.}}: Bare Demo of IEEEtran.cls for IEEE Journals}
%



	
 \author{Wei Wei,~\IEEEmembership{Graduate Student Member,~IEEE,}
	Zheng~Lin,~\IEEEmembership{Graduate Student Member,~IEEE,}
    Tao~Li,~\IEEEmembership{Graduate Student Member,~IEEE,}
    Xuanheng~Li,~\IEEEmembership{Member,~IEEE,}
	and~Xianhao~Chen,~\IEEEmembership{Member,~IEEE}

\thanks{The work was supported in part by the Research Grants Council of Hong Kong under Grant 27213824 and in part by HKU IDS Research Seed Fund under Grant IDS-RSF2023-0012.
\textit{(Corresponding author: Xianhao Chen)}
}
\thanks{ Wei Wei, Zheng Lin, Tao Li and Xianhao Chen are with the Department of Electrical and Electronic Engineering, University of Hong Kong, Pok Fu Lam, Hong Kong SAR, China. Xianhao Chen is also with HKU Musketeers Foundation Institute of Data Science, University of Hong Kong, Pok Fu Lam, Hong Kong SAR,  China. (e-mail: weiwei@eee.hku.hk; linzheng@eee.hku.hk; 
lthku999@connect.hku.hk;
xchen@eee.hku.hk).
}
 \thanks{Xuanheng Li is with the School of Information and Communication Engineering, Dalian University of Technology, Dalian 116023, China (e-mail: xhli@dlut.edu.cn).}
}

\maketitle

\begin{abstract}
To support large-scale model training, split learning (SL) enables multiple edge devices/servers to share the intensive training workload. However, most existing works on SL focus solely on two-tier model splitting. Moreover, while some recent works have investigated the model splitting and placement problems for multi-hop SL, these solutions fail to overcome the resource idleness issue, resulting in significant network idle time. In this work, we propose a pipelined SL scheme by addressing the joint optimization problem of model splitting and placement (MSP) in multi-hop edge networks. By applying pipeline parallelism to SL, we identify that the MSP problem can be mapped to a problem of minimizing the weighted sum of a bottleneck cost function (min-max) and a linear cost function (min-sum). Based on graph theory, we devise a bottleneck-aware shortest-path algorithm to obtain the optimal solution. Besides, given the MSP outcomes, we also derive the closed-form solution to the micro-batch size in the pipeline. Finally, we develop an alternating optimization algorithm of MSP and micro-batch size to solve the joint optimization problem to minimize the end-to-end training latency. Extensive simulations have demonstrated the significant advantages of our algorithm compared to existing benchmarks without pipeline parallelism.
\end{abstract}

\begin{IEEEkeywords}
Mobile edge computing (MEC), split learning, wireless multi-hop network, pipeline parallelism, combinatorial optimization problem (COP).
\end{IEEEkeywords}

%
\IEEEpeerreviewmaketitle

\section{Introduction}

{Training AI models in the cloud is considered the status quo paradigm, which harnesses the power of large-scale computing units to train powerful AI models. However, considering the fact that most data-generating devices are located at the network edge, cloud training faces critical challenges, including excessive data transmission latency ($30 \times $ to $ 100 \times$ greater than that experienced in edge environments \cite{8016573}), high bandwidth costs, scalability issues, and data privacy concerns \cite{lin2025hierarchical,hu2024accelerating,lin2024splitlora}.} As a result, edge learning becomes an emerging training paradigm that leverages computing capabilities located close to data sources for AI model training~\cite{ZW_spectrum}. By minimizing privacy exposure and data transmission costs, edge learning is expected to play a vital role in various mission-critical applications such as healthcare~\cite{9309177,mehandru2024evaluating,tang2024merit}, industrial automation~\cite{9139976,ZW2024ultra-LoLa}, and smart-city operations~\cite{7876955,lin2025leo}, where directly sharing personal, enterprise, or government data is bandwidth-costly or prohibited by data protection regulations \cite{voigt2017eu}.


However, as AI models expand in size, it becomes increasingly challenging and time-consuming to train advanced AI models on a \textit{single} edge server~\cite{lin2024split}; for example, state-of-the-art on-device models, such as TinyLLaMa, consist of over 1.1 billion parameters \cite{zhang2024tinyllama,10835069}, impose substantial training workload that can easily overwhelm an edge server, resulting in \textit{excessive training latency} and \textit{memory overflow}. To tackle the above challenges, \textit{multi-hop} split learning (SL) can be employed to facilitate the training of giant models at the network edge~\cite{gupta2018distributed,vepakomma2018split,ha2021secure}. The initial motivation behind vanilla SL is to train a model in a privacy-enhancing manner, which is achieved by splitting a model between a server and a user for training~\cite{gupta2018distributed,10529950}. By extending it to multi-hop SL, a large AI model can be partitioned into multiple submodels and distributed across geographically dispersed edge servers. In practice, these edge servers can either be enterprise private edge servers or public edge servers co-located with base stations in mobile networks. Regardless of their types, when forming a mesh of computing facilities, they constitute a distributed yet powerful computing infrastructure for large-scale model training \cite{HuaweiTech2023}. As shown in Fig. \ref{fig:pipeline_comp}(a), splitting a VGG-16 model among multiple edge servers can significantly reduce the training latency. Moreover, model splitting can also decrease the memory demands per server, which is equally crucial as the memory space is one of the most scarce resources on edge servers/devices.

\begin{figure}[t!]

 \subfigure [\centering{Total latency of pipelined SL versus the number of servers.}]{
	{

  \includegraphics[width=3.89375cm]{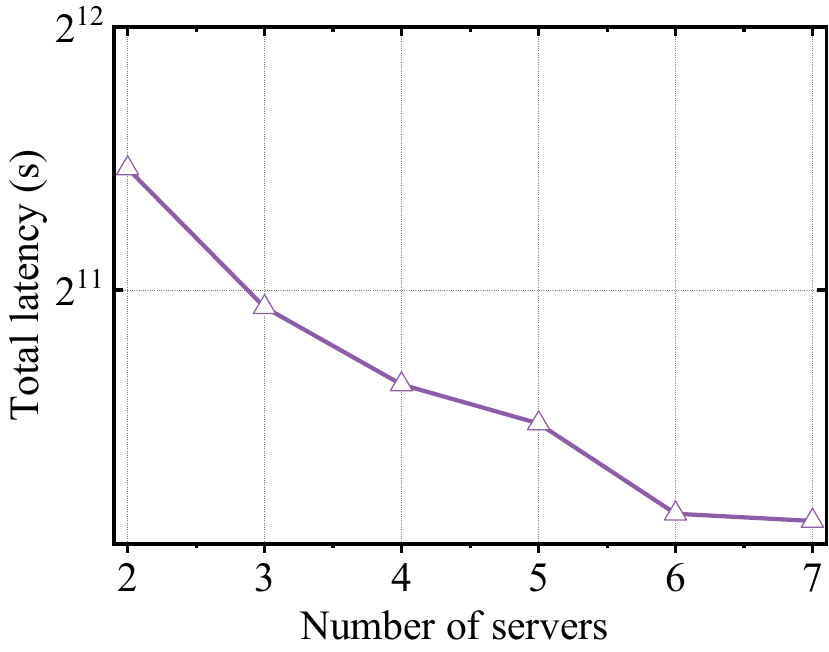}
			}
 }
 \subfigure [\centering{Comparison between pipelined SL and SL without pipeline parallelism.}]{
	{
\includegraphics[width=3.89375cm]{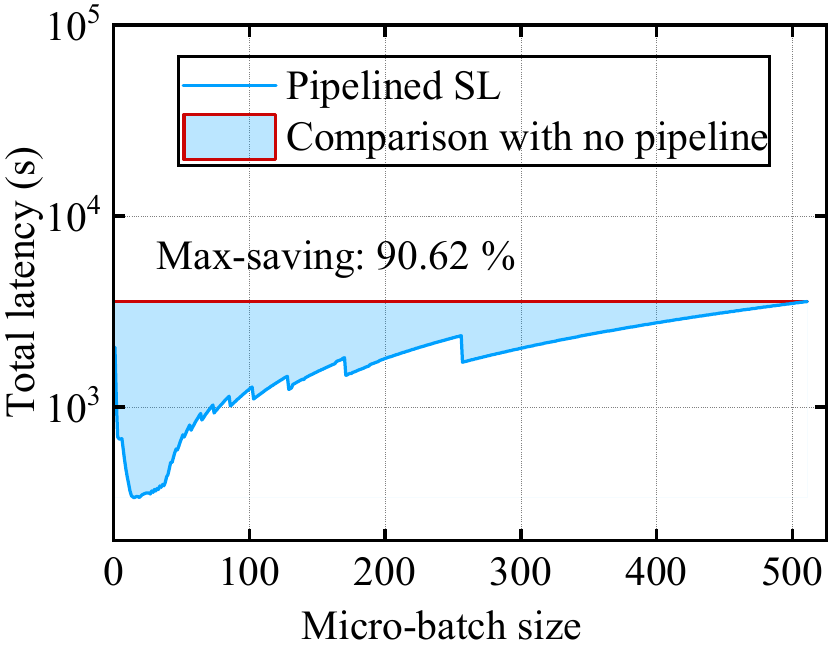}
	}
 }
\caption{Total latency in pipelined SL, which trains a VGG-16 model~\cite{simonyan2014very} on the CIFAR-10 dataset~\cite{NIPS2012_c399862d} under  IID settings. The computing capability of each edge server (NVIDIA GeForce RTX 4090) is uniformly distributed within [1, 10] TFLOPS. The total available bandwidth ranges from 10 MHz to 200 MHz.} 
		\label{fig:pipeline_comp}
\end{figure}

While multi-hop SL provides a promising solution to train a large-sized model at the network edge, implementing this framework encounters two major difficulties. The first difficulty comes from \textit{resource heterogeneity}. Since model training in edge networks involves servers with significantly heterogeneous computing, memory, and communication capabilities~\cite{10835069,khouas2024trainingmachinelearningmodels}, inappropriate model splitting and placement can result in bottleneck servers/links, leading to intolerable latency. For instance, a typical edge server might have a GPU like the NVIDIA GeForce RTX 4090 with 82 TFLOPS, while an edge device like the Raspberry Pi 4 performs at just 13.5 GFLOPS, leading to a performance difference of nearly four orders of magnitude~\cite{RaspberryPiSpecs,NVIDIATeslaSpecs}.
The second difficulty stems from \textit{resource idleness}. Given the interdependency of forward and backward passes in multi-hop SL, the majority of servers/links sit idle during training, resulting in excessive latency. This phenomenon becomes even severer as the number of hops increases, hampering the scalability of training large-sized models in an edge computing network.

Despite the fact that resource optimization of multi-hop SL can significantly influence its system performance, none of the existing schemes address resource heterogeneity and resource idleness as a whole. On the one hand, prior works on distributed training usually consider multi-GPU or multi-server training in cloud centers, assuming \textit{homogeneous} GPU configurations and \textit{identical} inter-GPU communication latency~\cite{huang2019gpipe,10.1145/3341301.3359646}. As a result, 
model splitting and placement are normally not optimized therein. However, model splitting and placement under heterogeneous communication-computing resource constraints are essential considerations in edge networks. 
On the other hand, a few recent works have addressed model splitting and placement in edge networks by considering heterogeneous communication and computing capabilities. Nevertheless, since these works mostly focus on multi-hop split inference problems~\cite{9562523,zw_AIOutage,zhiyan_ISEA_survey} or consider a single batch of training data successively going through the forward and backward passes, none of them mitigate resource idleness issues, as only one server or communication link can be active at a time during the process.

To address all the above issues, in this paper, we propose a pipelined SL scheme as our answer to effective implementations of multi-hop SL in edge networks. Unlike existing multi-hop split learning or inference schemes~\cite{9562523,10.1145/3037697.3037698,8737614,9984691}, our proposed framework involves an essential pipeline parallelism procedure, which divides a training dataset into multiple micro-batches for parallel processing across edge servers~\cite{liang2023survey}. As illustrated in Fig. \ref{fig:pipeline_comp}(b), compared with SL without pipeline parallelism, the proposed scheme can considerably reduce training latency by mitigating network idleness upon forward and backward passes. Based on this framework, we study the optimization problem of model splitting and placement (MSP)  for pipelined SL in edge networks with limited communication-computing resources. To minimize the end-to-end latency, we discern that the resulting MSP problem can be mapped to a problem of minimizing the weighted sum of a bottleneck cost function (min-max) and a linear cost function (min-sum), which can be solved exactly through graph theory. Specifically, the linear cost stems from the latency for the first micro-batch to be finished, whereas the bottleneck cost comes from the pipeline process. A block coordinate descent (BCD) is then developed to solve the joint problem by also incorporating the optimization of micro-batch size. Our key contributions are summarized below.

\begin{itemize}
\item To our knowledge, this is the first work that formulates a unified optimization framework that jointly determines model splitting \& placement (MSP) and micro-batch size for pipelined split learning in multi-hop edge networks, with the explicit goal of minimizing end-to-end training latency under heterogeneous communication, computation, and memory constraints.

\item We show that the resultant MSP subproblem can be mapped to a problem with combined min-max and min-sum objective. By mapping the problem into a graph, we develop a \textit{bottleneck-aware shortest-path} algorithm to obtain the optimal solution.

\item Given the fixed MSP results, we obtain the optimal closed-form solution to micro-batch size for the micro-batching subproblem. Subsequently, we devise a BCD-based method to obtain an efficient sub-optimal solution to the joint optimization problem.

\item Simulations have shown the significant advantages of our approaches compared with existing multi-hop SL benchmarks.
\end{itemize}

The remainder of this paper is organized as follows. Section~\ref{Rel_Work} introduces related work. Section~\ref{system_model} elaborates on the system model. We formulate the optimization problem in Section~\ref{formulation}, develop the corresponding solution approach in Section~\ref{solution}, and provide the simulation results in Section \ref{simulations}. Finally, we conclude the paper in Section~\ref{conclusion}.


\section{Related Work}\label{Rel_Work}
Several existing works have studied multi-hop split machine learning, including split inference and learning, under computing, memory, and/or communication constraints. Kang et al. \cite{10.1145/3037697.3037698,9944188, 10298247,9839238,10193767} first exploit the potential of distributing DNN inference between cloud servers and mobile devices to optimize latency, energy consumption, and throughput. Specifically, a lightweight scheduler has been introduced to partition DNN tasks in cloud-edge systems dynamically. In \cite{9562523}, a scheme called HiveMind has been devised to facilitate multi-hop split inference in 5G networks. By reformulating the model splitting problem as a min-cost graph search, HiveMind efficiently adapts to real-time network dynamics, significantly reducing signaling overhead.  Moreover, Hu et al.\cite{8737614} optimally partition DNNs between edge and cloud computing based on dynamic network conditions, which addresses the computational and transmission trade-offs. Xiao et al.
\cite{9984691}  propose an efficient multi-hop edge inference, which allows mobile devices to dynamically select the optimal partition point of deep learning models and choose collaborative edge servers based on real-time conditions. However, although split inference and SL bear similarities, some significant differences exist. Particularly, split inference often takes one data sample at a time rather than processing a batch of training data. This implies that pipeline parallelism, considered essential for the SL process, does not appear in split inference. Mathematically, our work will show that the latency minimization problem of pipelined SL turns out to be a min-max-min-sum problem, as the total latency accounts when the first micro batch enters the multi-hop networks until the last micro batch comes out.

Pipeline parallelism is a widely studied topic in distributed training~\cite{10.1145/3458817.3476145,huang2019gpipe,kosson2021pipelinedbackpropagationscaletraining}. The basic idea  is to make computing and communication ``interleaved'' across different devices (often GPUs) to reduce idle time. PipeDream~\cite{10.1145/3341301.3359646} uses a one forward pass followed by one backward pass (1F1B) scheduling approach to interleave forward pass (FP) and backward pass (BP) and addresses weight staleness issues arising from asynchronous backward updates. GPipe~\cite{huang2019gpipe} uses synchronous stochastic gradient descent (SGD) and divides mini-batches into smaller micro-batches to minimize bubble time. Synchronization happens at the end of each mini-batch by aggregating gradients from all micro-batches. Moreover, Chimera \cite{10.1145/3458817.3476145} employs bidirectional pipelines for efficient training of large-scale neural networks, which minimizes the number of pipeline bubbles by up to 50\% compared to GPipe.
Nevertheless, this line of research does not consider the MSP problem under heterogeneous communication-computing constraints -- a feature that edge networks possess.

Multi-hop SL has also been investigated in some literature.  In \cite{10.1145/3538641.3561500}, RNN is partitioned into sub-models and deployed on multiple mobile devices for training via inter-device communication. DAPPLE \cite{10.1145/3437801.3441593} combines data and pipeline parallelism for synchronous training of large DNN models. It features a dynamic programming-based planner to split and place model layers on devices. In \cite{tian2021jmsnasjointmodelsplit}, Tian et al. devise a model splitting and neural architecture search framework for SL in a mesh network. However, pipeline parallelism has not been considered in these works. The most related work is \cite{tirana2024mpslmultihopparallelsplit}, which proposes a parallel SL framework addressing multi-hop MSP problem.  Nevertheless, this framework  does not account for communication latency in the MSP optimization problem, overlooking the effects of network topology and bandwidth. In addition, it employs a fixed batch size throughout training without taking advantage of joint batching and MSP optimization.  To our best knowledge, our paper presents the first pipelined SL framework to address the MSP optimization under \textit{communication-computing} resource constraints.

\begin{table}[t]
{
	\centering\caption{Summary of Key Notations.}
	{\footnotesize
	\setlength\tabcolsep{3pt}
	\renewcommand\arraystretch{1}
	\newcommand{\tabincell}[2]{\begin{tabular}{@{}#1@{}}#2\end{tabular}}%
	\begin{threeparttable}
		\begin{tabular}{c|c}
			\bottomrule
			\bottomrule
			{\tabincell{c}{\textbf{Notations}}} & {\tabincell{c}{\textbf{Descriptions} }}  \\ 
			\hline
			$I$ & The number of layers of a neural network   \\
			{\tabincell{c}{\textbf{$K$} }} & The number of submodels  \\
			$N$ & The number of edge servers   \\
  $M$ & The number of clients   \\
			{\tabincell{c}{\textbf{$x_{ik}$} }} & {\tabincell{c}{
   The cutting strategy, where $x_{ik}=1$ represents the last layer  \\ of the $k$-th submodel with the index $i$ and $x_{ik}=0$ otherwise
         }} \\
			{\tabincell{c}{\textbf{$y_{kn}$} }} & {\tabincell{c}{The placement strategy, where $y_{kn}$ =1 if the $k$-th  submodel  \\ is deployed to  client/server $n$ and 0 otherwise }} \\
			{\tabincell{c}{\textbf{$B_m$} }} & The number of  total data samples  from client $m$   \\
			{\tabincell{c}{\textbf{$B$} }} & The number of   total data samples  in the mini-batch   \\
			{\tabincell{c}{\textbf{$b$} }} & The number of server-side data samples in a micro-batch   \\
            	{\tabincell{c}{\textbf{$b_m$} }} & The number of client-side data samples in a micro-batch   \\
			{\tabincell{c}{\textbf{$D_k$} }} & The data size of activations   \\
			{\tabincell{c}{\textbf{$D_k^\prime$} }} & The data size of activations' gradients   \\
			{\tabincell{c}{\textbf{$\eta_k$} }} & The memory cost to process the $k$-th
submodel   \\
			{\tabincell{c}{\textbf{$M_n$} }} & The maximum GPU memory of client/server $n$   \\
			{\tabincell{c}{\textbf{$P$} }} & The  number of parameters   \\
			{\tabincell{c}{\textbf{$W_{nn^{\prime}}$} }} & The channel bandwidth between $n$ and $n^\prime$   \\
			{\tabincell{c}{\textbf{$N_0$} }} & The noise power  \\
			{\tabincell{c}{\textbf{$p_n$} }} & The transmit power of client/server $n$   \\
			{\tabincell{c}{\textbf{$d_{nn^{\prime}}$} }} & The distance between  $n$ and  $n^{\prime}$   \\
			{\tabincell{c}{\textbf{$\gamma$} }} &  The path loss exponent   \\
	{\tabincell{c}{\textbf{$r_{nn^{\prime}}$} }} &  {\tabincell{c}{The achievable data rate over wireless link \\ between  $n$ and $n^{\prime}$}}  \\
	{\tabincell{c}{\textbf{$f_n$} }} &  The computing capability of client/server $n$\\
	{\tabincell{c}{\textbf{$\kappa_n$} }} &  The computing intensity of client/server $n$\\
	{\tabincell{c}{\textbf{$w_i$} }} & {\tabincell{c}{The computing workload per data sample \\of FP for the first $i$-th layers} } \\
	{\tabincell{c}{\textbf{$\rho_i$} }} & {\tabincell{c}{The computing workload per data sample \\of BP for the first $i$-th layers }} \\
	{\tabincell{c}{\textbf{$\varphi_i$} }} & The size of activations of the $i$-th layer  \\
	{\tabincell{c}{\textbf{$\phi_i$} }} & The size of activations' gradients of the $i$-th layer  \\
   \toprule
		\end{tabular}%
	\end{threeparttable}
}
	\label{notation}%
    }
\end{table}%
\section{System Model}\label{system_model}
\subsection{Architecture Overview}
Fig. \ref{multi-system} illustrates a multi-hop SL framework based on pipeline scheduling with multiple edge nodes including clients and edge servers. In practice, the next-generation user plane function (UPF) enables the routing of application traffic among users and mobile edge computing (MEC) servers associated with different network entities, thereby facilitating a multi-hop SL paradigm~\cite{9562523}. Each edge server, e.g., an MEC node in the platform, can execute FP and BP model trainings. Moreover, we consider a scenario where  $N$ edge servers $\bm S = \{ 1,  2, ..., N \}$  participate in the multi-hop SL system to process a neural network with $I$ layers cooperatively. 
{We define the overall set of nodes as $\mathcal{N}$ and denote the set of clients as $\mathcal{N}_c$, where $\mathcal{N}_c \subset \mathcal{N}$.}
Next, we define the decision variables, i.e., model split matrix and model placement matrix as follows,
\begin{itemize}
    \item Model split matrix $\bm x=\{x_{ik}|\ i\in[1,I],\ k\in[1,K-1]\}$: $x_{ik}=1$ represents that the last layer of the $k$-th submodel is layer $i$ and $x_{ik}=0$ otherwise.
    \item  {Model placement matrix $\bm y=\{y_{kn}|\ k\in[1,K],\ n\in\mathcal{N}\}$: For edge server $n\in\mathcal{N}\setminus \mathcal{N}_c$, $y_{kn}=1 
    $ denotes that the $k$-th submodel $(k\in[2,\ K])$ is deployed on server $n$, and $y_{kn}=0$ otherwise. 
    For client  $n\in\mathcal{N}_c$, we enforce $y_{kn}=1 
    $, indicating that the first submodel $(k=1)$ should be deployed on the client side for preventing raw data transmissions to edge servers.}
\end{itemize}

Herein, $K$ denotes the maximum number of submodels after model splitting. Notably, the copies of the first submodel ($k=1$) are placed on $M$ clients  $\bm C = \{ 1,  2, ..., M \}$. In each training round, client $m$ processes mini-batch $B_m$ of data samples. {
 Assume that the clients connect to an access edge server, the edge server draws mini-batch $B=\sum_{m=1}^M B_m$ of data samples for split training based on mini-batch SGD.
As shown in Fig. \ref{multi-system}, the mini-batch $B$  will be divided into smaller \textit{micro-batch} $b$  for pipeline parallelism,  thus enabling parallel process across edge nodes.} Besides, we assume that the average data rate between servers and the network topology remain unchanged within each training round. 
Moreover, we mainly focus on one training round while omitting the training round index. 
 Our optimization goal is to train the neural network with the minimum end-to-end latency by jointly optimizing the cut
layers, submodel placement, and micro-batching. The training process in each training round is elaborated as follows. 

\begin{figure}[t!]
\centering
\includegraphics[width=9cm]{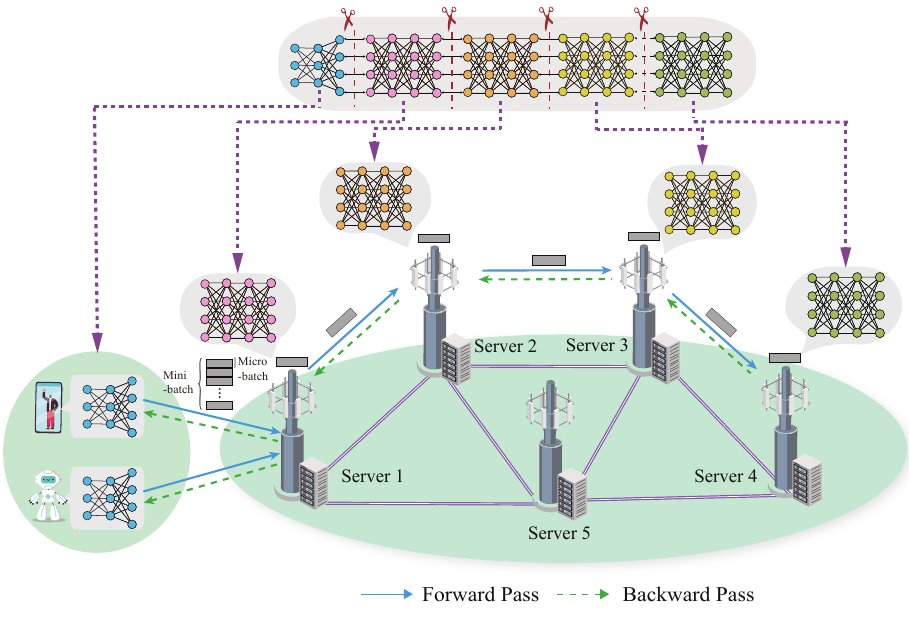}
	\flushleft
\caption{Pipelined split learning in multi-hop edge networks.}
\label{multi-system}
\end{figure}



\subsection{FP \& Activation Transmissions}
\subsubsection{The FP process} 
In the FP process, submodel $k$ (placed on a client or a client or an edge server) receives the activations of the prior submodel (if any) as the inputs, performs FP to produce the output activations, and feeds the activations to the next submodel, which is placed on another edge server. The process terminates once the computation of all layers of the entire neural network is finished. {Without loss of generality, we assume each client $m$ draws $b_m$ data samples and the first server draws a micro-batch of $b=\sum_{m=1}^M b_m$ data samples to execute the  FP process. Specifically,  we arrange 
$
\left\lfloor \frac{b}{M} \right\rfloor
$
data samples to each of the first \(M-1\) clients and allocate the remaining
$
b - (M-1)\left\lfloor \frac{b}{M} \right\rfloor
$
data samples to the \(M\)-th client, which is expressed by
\begin{align}
b_m =\left\{\begin{aligned}
&\left\lfloor {b/M} \right\rfloor,\  1 \le m \le M-1, \\
&b - (M-1)\left\lfloor {b/M} \right\rfloor,\  m = M.
\end{aligned}\right.  
 \end{align}

In this way, the forward pass latency of edge server $n$ for processing the $k$-th submodel can be expressed as
\begin{align}
t_{k,n}^{\mathrm{F}}(b,\boldsymbol{x})=\left\{\begin{aligned}
	&b_m
{\kappa _n\,\delta _{kn}^{\mathrm{F}}(\boldsymbol{x})}/{f_n}+t_{0}^c
,\ \forall n\in \mathcal{N} _c,\\
&{b\,\kappa _n\,\delta _{kn}^{\mathrm{F}}(\boldsymbol{x})}/{f_n}+t_{0}^s
,\ \forall n\in \mathcal{N}\setminus \mathcal{N}_c,\\
\end{aligned}\right.  
 \end{align}
where $\boldsymbol x$ is the collection of $x_{ik}$,  $f_n$ denotes the computing capability (in CPU/GPU frequency) of server $n$,  $\kappa_n$ represents the computing intensity \cite{lin2023efficient},   $t_0^c$ and $t_0^s$  are coefficients related to initialization and model loading,
and
  \begin{align}
&\delta_{kn}^{\mathrm{F}}(\boldsymbol{x})= \notag\\&\left\{\begin{aligned}
&\sum\limits_{i=1}^{I} x_{ik} \, w_i,  \ \forall n \in \mathcal{N}_c ,\ k=1,\\
\displaystyle &\sum\limits_{i=1}^{I} \left(x_{ik}-x_{i(k-1)}\right) w_i,\ \forall n\in \mathcal{N}\setminus \mathcal{N}_c,\ k\in[2,K],\\
\end{aligned}\right. 
\end{align}
with $w_i$ being the computing workload per data sample of FP for the first $i$-th layers. }


\subsubsection{Activations transmission} 
After the $k$-th submodel completes the FP process on the client $m$ or edge server $n$, the host client/server transmits the activations to the subsequent server $n^\prime$ over a wired/wireless channel. For instance, if the connection is wireless (e.g., using 5G wireless backhaul), the achievable data rate $r_{nn^{\prime}}$ over the link is given by
\begin{align}
r_{nn^{\prime}}&=W_{nn^{\prime}}\cdot \log \left( 1+{{(d_{nn^{\prime}})}^{-\gamma}p_n}/{N_0} \right),\notag \\
& \quad \quad \quad \quad \quad \quad \quad \quad \quad \quad  \forall n \in  \mathcal{N} , \ \forall n^\prime \in \mathcal{N} ,  
\end{align}
where $p_n$ is the transmit power, $d_{nn^{\prime}}$ denotes the distance between edge nodes $n$ and  $n^\prime$, $W_{nn^{\prime}}$ is the channel bandwidth, $\gamma$ is the path loss exponent, and $N_0$ is the noise power. Besides, the data size $D_{k}(\boldsymbol{x},b)$ of the activations generated by the $k$-th submodel   is defined as
\begin{align}
D_k(\boldsymbol{x},b)=\left\{ \begin{aligned}
	&b_m  \sum\limits_{i=1}^{I}x_{ik}\varphi_i,\ k=1,\\
	&b  \sum\limits_{i=1}^{I}x_{ik}\varphi_i,\ \forall k\in[2,K-1],\\
\end{aligned} \right. 
\end{align}
where $\varphi_i$ stands for the size of activations of at layer $i$.
Thus, the communication latency can be expressed as 
\begin{align}
t_{k,nn^{\prime}}^{\mathrm{F}}=
{D_k(\boldsymbol{x},b)}/{r_{nn^{\prime}}},\ \forall n,\ n^\prime\in \mathcal{N} , \ \forall k\in [1,K-1].
\end{align}
\subsection{BP \& Activations' Gradients Transmission}
\subsubsection{The BP process} 
After a micro-batch finishes the FP process,  the BP process begins at the last layer and moves backward to the very first layer. Each submodel $k\ (k\neq K)$ receives the activations' gradients of subsequent submodel as the inputs, performs computations, and feeds the results to the preceding submodel. The process terminates once the computation of all layers is finished.
The backward pass latency  $t_{k,n}^{\mathrm{B}}(b,\boldsymbol{x})$ on server $n$ for processing the $k$-th submodel is expressed as
\begin{align}
& t_{k,n}^{\mathrm{B}}(b,\boldsymbol{x})=\left\{\begin{aligned}
&t_1^c, \ \forall n \in \mathcal{N}_c,\ 0 < b_m\leqslant b_{\mathrm{th}}^c,   \\
	&t_1^s, \ \forall n\in \mathcal{N}\setminus \mathcal{N}_c,\ 0 < b\leqslant b_{\mathrm{th}}^s,\\
    &(b_m-b_{\mathrm{th}}^c)\kappa _n\delta _{kn}^{\mathrm{B}}(\boldsymbol{x})
 /f_n+t_1^c, \\
 & \quad \quad \quad \quad \quad \quad   \ \forall n \in \mathcal{N}_c,\  b_{\mathrm{th}}^c<b_m \leqslant B, \\
 &(b-b_{\mathrm{th}}^s)\kappa _n\delta _{kn}^{\mathrm{B}}(\boldsymbol{x})
 /f_n+t_1^s, \\
 & \quad \quad \quad \quad \quad \quad  \ \forall n\in \mathcal{N}\setminus \mathcal{N}_c,\  b_{\mathrm{th}}^s<b \leqslant B,\\
\end{aligned}\right. 
\end{align}
with
\begin{align}
&\delta_{kn}^{\mathrm{B}}(\boldsymbol{x}) =\notag\\ &\left\{\begin{aligned}
& \sum_{i=1}^{I} x_{ik}\, \rho_i,  \ \forall n \in \mathcal{N}_c,\ k=1,\\
& \sum_{i=1}^{I} \bigl( x_{ik} - x_{i(k-1)} \bigr)\, \rho_i,  \ \forall n \in \mathcal{N}\setminus\mathcal{N}_c,\ k\in[2,K], \\
\end{aligned}\right. 
\end{align}
where $\rho_i$ is the computing workload per data sample of BP for the first $i$-th layers, and the coefficients $b_{\mathrm{th}}^c$, $b_{\mathrm{th}}^s$, $t_1^c$ and $t_1^s$  are determined by specific DNN models and hardware~\cite{ 9252924}.
\subsubsection{Transmissions of activations' gradients} When the BP process of a submodel is completed, activations’ gradients at the cut layer will be transmitted through the communication links. The data size $D_{k}^\prime(\boldsymbol{x},b)$ of the activations' gradients generated by the $k$-th submodel 
   is expressed as
\begin{align}
D_{k}^{\prime}(\boldsymbol{x},b)=\left\{ \begin{aligned}
	&b_m  \sum\limits_{i=1}^{I}x_{ik}\phi_{(i+1)},\ k=1,\\
	&b  \sum\limits_{i=1}^{I}x_{ik}\phi_{(i+1)},\ \forall k \in [2,K-1],\\
\end{aligned} \right. 
\end{align}
where  $\phi_{(i+1)}$ stands for the size of activations' gradients at cut layer $i+1$. Thus, 
the communication latency between node $n^{\prime}$ and $n$ for the $k$-th submodel can be expressed as~\cite{lin2024adaptsfl}
\begin{align}
t_{k,n^{\prime}n}^{\mathrm{B}}=
{D_k^\prime(\boldsymbol{x},b)}/{r_{n^{\prime}n}},\ \forall n, \ n^\prime\in \mathcal{N} , \ \forall k\in [1,K-1].
\end{align}
\subsection{Memory Consumption}
In multi-hop SL, it is crucial to avoid out-of-memory failures during model training~\cite{9428512}.
The required memory space to process the $k$-th submodel is denoted by
\begin{align}
 &\eta_k(\boldsymbol{x},b)= \notag\\
 &\left\{
\begin{aligned}
   &b_m\sum\limits_{i=1}^Ix_{ik}(
 {\widetilde \varphi} _i
 +\widetilde{\phi}_i + {\widetilde \sigma}_i + \beta_i),\ k=1,\\
 & b  \sum\limits_{i=1}^{I}(x_{ik}-x_{i(k-1)})(
 {\widetilde \varphi} _i
 +\widetilde{\phi}_i + {\widetilde \sigma}_i + \beta_i),\ k\in[2,K],  \\
\end{aligned}
\right.
\end{align}
where ${{\widetilde \varphi} _i} = \sum\limits_{i' = 1}^i {{\varphi _{i'}}} $ and ${{\widetilde \phi}_i} = \sum\limits_{i' = 1}^i {{{\phi}_{i'}}} $ represent the cumulative sum of the data size of activations and activations' gradients for first $i$ layers of neural network, ${{\widetilde \sigma} _i}$ is the data size of the optimizer state for the first $i$ layers of neural network, depending on the choice of the optimizer (e.g. SGD, Momentum, and Adam), and $\beta_i$ is the parameter size of the first $i$ layers.

\section{Problem Formulation}\label{formulation}
In our SL design, the core innovation lies in pipelining successive micro-batches under heterogeneous communication-computing resource constraints, thereby mitigating resource heterogeneity and idleness issues. However, this \textit{pipeline process in edge networks} also leads to a challenging optimization problem different from previous works~\cite{9562523,tirana2024mpslmultihopparallelsplit}. Specifically, we aim to minimize the total training latency of pipelined SL, which contains 1)  latency $T_f(\boldsymbol{x}, \boldsymbol{y},b)$ for the first micro-batch to complete the FP and BP processes, and 2) the pipeline latency for the remaining micro-batches to finish the processes. For the first part, $T_f(\boldsymbol{x}, \boldsymbol{y},b)$ is given by
\begin{align}
&T_f(\boldsymbol{x}, \boldsymbol{y},b)=\max_{n \in \mathcal{N}_c}\{t_{1,n}^{\mathrm{F}}(b,\boldsymbol{x}) + \sum\limits_{n^\prime=1}^Ny_{2n^\prime}t_{1,nn^{\prime}}^{\mathrm{F}}(\boldsymbol{x},b)\}
+ \sum_{k=2}^K\notag\\
&\sum_{n=1}^{N}{y_{kn}}\big[t_{k,n}^{\mathrm{F}}(b,\boldsymbol{x}) +t_{k,n}^{\mathrm{B}}(b,\boldsymbol{x})\big]+ \sum_{n=1}^{N}\sum_{n' = 1 }^{N}
\big\{\sum_{k=2}^{K-1}
{y_{kn}}y_{(k+1) n^{\prime}}\notag\\
&{t_{k,nn^{\prime}}^{\mathrm{F}}(\boldsymbol{x},b)}+\sum_{k=2}^{K-1}y_{(k+1) n^{\prime}}{y_{kn}}{t_{k,n^{\prime}n}^{\mathrm{B}}(\boldsymbol{x},b)} \big\}+\max_{n \in \mathcal{N}_c}\{t_{1,n}^{\mathrm{B}}(b,\boldsymbol{x})
\notag\\
&  + \sum\limits_{n^\prime=1}^Ny_{2n^\prime}t_{1,n^{\prime}n}^{\mathrm{B}}(\boldsymbol{x},b)\},\label{Tf}
\end{align}
where $\boldsymbol{x}$ is the collection of $x_{ik}$ and $\boldsymbol{y}$ is the collection of $y_{kn}$. The four terms in (\ref{Tf}) correspond to the FP and activation transmissions on clients, the FP and BP on edge servers, the communication latency among edge servers, and the BP and activations' gradients transmissions on clients, respectively.
 
Considering multiple micro-batches to be processed, we assume that  ``FP \& activations transmission'' and  ``BP \& activations’ gradients transmission'' are executed in \textit{pipeline}. To derive the pipeline latency, we need to calculate the \textit{difference of completion time of two consecutive micro-batches}, which can be derived from
\begin{equation}
\begin{aligned}
    &T_i(\boldsymbol{x}, \boldsymbol{y}, b)= \max\bigg\{\max_{k\in[2,K-1]} \big\{ \sum\limits_{n=1}^N   {{{y_{kn}}{t_{k,n}^{\mathrm{F}         
       }(b,\boldsymbol{x})}}},\sum\limits_{n=1}^N{y_{kn}}\\
  &   {t_{k,n}^{\mathrm{B}}(b,\boldsymbol{x})},\sum\limits_{n=1}^N\sum_{{n' = 1 }}^{N}{y_{kn}}y_{(k+1) n^{\prime}} {t_{k,nn^{\prime}}^{\mathrm{F}}(\boldsymbol{x},b)},
\sum\limits_{n=1}^N\sum_{{n' = 1 }}^{N}{y_{kn}}\\
  & y_{(k+1)n^{\prime}}{t_{k,n^{\prime}n}^{\mathrm{B}}(\boldsymbol{x},b)}
\big\},\max_{n \in \mathcal{N}_c}\big\{t_{1,n}^{\mathrm{F}}(\boldsymbol{x})+\sum\limits_{n^\prime=1}^Ny_{2n^\prime}t_{1,nn^{\prime}}^{\mathrm{F}}(\boldsymbol{x},b)\big\},
\\&   \max_{n \in \mathcal{N}_c}\big\{t_{1,n}^{\mathrm{B}}(\boldsymbol{x})+ \sum\limits_{n^\prime=1}^Ny_{2n^\prime}t_{1,n^{\prime}n}^{\mathrm{B}}(\boldsymbol{x})\big\}\bigg\}.
\label{max-function}
\end{aligned}
\end{equation}

Intuitively, $T_i(\boldsymbol{x}, \boldsymbol{y}, b)$ is determined by the \textit{bottleneck latency} over the multi-hop systems, which is equal to the maximum latency over any computing server or communication link. 
Moreover, the batch per training round with $B$ data samples needs to be transmitted $\lceil{B}/{b}\rceil$ times. Thus, the total latency $L_t(\boldsymbol{x}, \boldsymbol{y}, b)$ can be formulated as
\begin{equation}
\begin{aligned}
L_t(\boldsymbol{x}, \boldsymbol{y}, b)
&= T_f(\boldsymbol{x}, \boldsymbol{y},b)+\lceil{(B-b)}/{b}\rceil T_i(\boldsymbol{x}, \boldsymbol{y}, b),
\label{per-round-latency}
\end{aligned}
\end{equation}
where the second term is the pipeline latency.
In the multi-hop SL system, \textit{the overarching goal is to jointly optimize the cutting layers, model placement, and micro-batch size to minimize the total latency $L_t(\boldsymbol{x}, \boldsymbol{y}, b)$}. Thus, the optimization problem can  be formulated as
	\begin{align}
	\mathcal{P}1:	&\underset{\{ \boldsymbol{x}, \boldsymbol{y}, b \}}{\mathrm{min}} \quad L_t(\boldsymbol{x}, \boldsymbol{y}, b) \notag \\
		\mathrm{s.t.}\quad
          & \mathrm{C1:}\quad  x_{ik}\in\{0,1\},\ \forall i \in [1,I], \ k \in [1,K-1],\notag \\
                  & \mathrm{C2:}\quad  y_{kn}\in\{0,1\},\ \forall k \in  [1,K], \ n \in \mathcal{N}, \notag \\
      & \mathrm{C3:}\quad  b\in \{1,2,...,B\}, \notag \\
        & \mathrm{C4:}\quad \sum_{i=1}^I x_{ik}=1,\ \forall k \in [1,K-1], \notag \\
          & \mathrm{C5:}\quad  \sum_{i=1}^{I^\prime }x_{ik} \leqslant \sum_{i=1}^{I^\prime }x_{i(k-1)} ,
          \notag \\
          & \quad \quad \quad  \quad \quad  \quad \quad \quad \quad 
 \  \forall k \in [2,K-1], \ I^\prime \in [1, I], \notag \\
          &\mathrm{C6:}	\quad  \sum_{n=1}^N y_{kn}=1,\ \forall k \in [2,K], \notag \\
  		&\mathrm{C7:}	\quad   0 \leqslant   m_{1}(x_{i1},b)\leqslant M_n,\ \forall n \in \mathcal{N}_c , \notag\\
    &\mathrm{C8:}	\quad   0 \leqslant  \sum_{k=2}^K y_{kn} m_{k}(x_{ik},b)\leqslant M_n,\ \forall n\in \mathcal{N}\setminus \mathcal{N}_c.
    \label{problem_p1}
	\end{align}

Constraints C1 and C2 ensure that model splitting and placement are binary variables. Constraint C3 means that the micro-batch size should be no more than the mini-batch size $B$. Since gradient updating will be made after each mini-batch, $B$ is a hyperparameter that can be appropriately chosen based on learning efficiency. Other constraints are explained below.

\begin{itemize}
 	\item Cutting layer constraints:  C4 ensures the uniqueness of cut layer selection for each submodel. Notice that C4 does not exclude the case that different submodels can select the same cut layer $i$, implying that the model can be split into less than $K$ parts, where $K \leqslant I$ is a predetermined constant.  This is because splitting models may not always lead to shorter latency due to extra communication latency.   Furthermore, the layer dependency Constraint C5 ensures that the neural network layers are in order, i.e., the index of the cut layer of the $(k-1)$-th submodel is no greater than that of the $k$-th submodel (the equality holds when there is an ``empty'' submodels).

  \item The placement constraint: C6 ensures that each submodel can be placed on one edge node. It is noted that submodel $k$ can be ``empty'' as alluded to earlier. 
\item The GPU memory constraint:  Constraint C7 and Constraint C8 ensure that the memory cost to process the $k$-th
submodel is no more than the maximum GPU memory capacity $M_n$.
\end{itemize} 


The optimization Problem $\mathcal{P}1$ turns out to be a nonlinear integer programming (NILP) problem since it involves the minimization of a non-linear function $L_t(\boldsymbol{x}, \boldsymbol{y}, b)$ and a non-linear Constraint C7. In what follows, we show Problem $\mathcal{P}1$ is NP-hard. 
\begin{proposition}
    The latency minimization Problem $\mathcal{P}1$ is NP-hard.
\end{proposition}

\begin{proof}
Problem $\mathcal{P}1$ can be decomposed into
three sub-problems, namely layer splitting, model placement, and micro-batching. When fixing the other two sub-problems, the placement sub-problem  is a knapsack problem~\cite{CACCHIANI2022105693} which is NP-hard. Consequently, Problem $\mathcal{P}1$ is NP-hard.
\end{proof}

\section{Solution Approach}\label{solution}
In this section, we first reformulate Problem $\mathcal{P}1$ and analyze the computation complexity of the latency minimization problem. Then, we decompose the problem into two subproblems, i.e., the MSP problem and the micro-batching problem. Given fixed MSP results, we derive the optimal solution to the micro-batching subproblem. Moreover, we identify that the MSP problem is a problem with a combined min-max min-sum objective function, which can be efficiently solved based on graph theory. 
 
\subsection{Problem Reformulation}
The optimization objective of Problem $\mathcal{P}1$ is nonlinear and contains the ceil function. To tackle this challenge, we introduce an auxiliary variable $T_1$ subject to $ T_1\ge T_i(\boldsymbol{x}, \boldsymbol{y}, b)$ and function  $\xi(b) =\lceil ({B-b}/{b})\rceil$.
$\mathcal{P}1$ is then converted into
	\begin{align}
	\mathcal{P}2:	&\underset{
	\{\boldsymbol{x}, \boldsymbol{y},b,T_1\}}{\min}T_f(\boldsymbol{x}, \boldsymbol{y},b)+\underset{\xi(b) }{\underbrace{\lceil ({B-b}/{b})\rceil  }} \cdot T_1 \notag  \\
		\mathrm{s.t.}\quad
          & \mathrm{C1} \text{ - } \mathrm{C8}, \notag \\
          & \mathrm{C9:}\  \sum_{k=2}^K y_{kn} t_{k,n}^{\mathrm{F}}(b,\boldsymbol{x})\leqslant T_1,\  \forall  n\in \mathcal{N}\setminus \mathcal{N}_c, \notag  \\
           &\mathrm{C10:}\   t_{1,n}^{\mathrm{F}}(b,\boldsymbol{x})\leqslant T_1,\    \forall n\in \mathcal{N}_c, \notag  \\
                              &  \mathrm{C11:}\   b{\sum_{k=2}^{K-1}{y_{kn}}y_{(k+1) n^{\prime}} \sum_{i=1}^{I}x_{ik}\varphi_i}/{r_{nn^{\prime}}}\leqslant T_1, \notag \\ &\ \ \ \ \ \  \ \forall    n,\ n^\prime\in\mathcal{N}\setminus \mathcal{N}_c, \notag  \\
                              &  \mathrm{C12:}\    b_m \sum_{i=1}^{I}x_{i1}\varphi_i/{r_{nn^{\prime}}}\leqslant T_1, \notag \\
                              &\ \ \ \ \ \  \ \forall   n\in\mathcal{N}_c,\ n^\prime\in\mathcal{N}\setminus \mathcal{N}_c, \notag  \\
                               & \mathrm{C13:}\  \sum_{k=2}^K y_{kn} t_{k,n}^{\mathrm{B}}(b,\boldsymbol{x}) \leqslant T_1,\  \forall  n\in\mathcal{N}\setminus \mathcal{N}_c, \notag  \\
                               &  \mathrm{C14:}\      t_{k,n}^{\mathrm{B}}(b,\boldsymbol{x})\leqslant T_1,\  \forall  n\in\mathcal{N}_c, \notag \\
                              &  \mathrm{C15:}\   b{\sum_{k=2}^{K-1}{y_{kn}}y_{(k+1)n^{\prime}} \sum_{i=1}^{I}x_{ik}\phi_{(i+1)}}/{r_{nn^{\prime}}}\leqslant T_1, \notag \\
                        &\ \ \ \ \ \ \  \forall   n,\ n^\prime\in\mathcal{N}\setminus \mathcal{N}_c,\notag \\
                              &  \mathrm{C16:}\     b_m{ \sum_{i=1}^{I}x_{i1}\phi_{(i+1)}}/{r_{nn^{\prime}}}\leqslant T_1, \notag \\
                              &\ \ \ \ \ \ \  \forall   n\in\mathcal{N}_c,\  n^\prime\in\mathcal{N}\setminus \mathcal{N}_c, 
                    \label{problem_p2}
	\end{align}

Specifically, the additional Constraints $\mathrm{C9}\text{-}\mathrm{C16}$ are established according to Eq. (\ref{max-function}). 
 The transformed Problem $\mathcal{P}2$ is equivalent to the original Problem $\mathcal{P}1$, as the optimal solution $T_1^*$ obtained from Problem $\mathcal{P}2$ should meet $T_1^*=T_i(\boldsymbol{x}, \boldsymbol{y}, b)$.

\subsection{Complexity Analysis}
If we take the brute-force approach and examine all possible layer splitting, model placement
and micro-batching options to tackle Problem $\mathcal{P}2$, then it requires calculating the latency for all $\sum\limits_{k=2}^N B\bigg ( \begin{array}{c}
	I-1\\
	k-1 \\
\end{array} \bigg)k !\bigg( \begin{array}{c}
	N\\
	k\\
\end{array} \bigg)$ options! For example, given $B=256$, the VGG-16 model to be split and deployed on no more 5 edge nodes have approximately $5.76 \times 10^{7}$ options. Examining such a huge number of options requires significant processing time. Moreover, many existing algorithms, such as branch and bound algorithm \cite{10.5555/2969033.2969194},  may not apply to a large-scale problem since the time complexity exponentially increases with the number of the integer variables, particularly with the multi-dimension resource allocation variables $\{x_{ik},\ y_{kn},\ b, \ T_1\}$. Instead of looking for search methods, in what follows, we discuss the more efficient solution approaches for the micro-batching problem and the MSP problem, respectively.


\subsection{Solution to the Micro-batching Problem}
By  fixing the decision variables $x_{ik}$, $y_{kn}$ and $T_1$ in $\mathcal{P}2$, the micro-batching problem can be expressed as
	\begin{align}
	\mathcal{P}3:\	&\underset{\{ b  \}}{\mathrm{min}}\  T_f(b)+\xi(b)  \cdot T_1 \notag \\
		\mathrm{s.t.}\quad
         & \mathrm{C3},\ \mathrm{C7}\text{ - }\mathrm{C16}, 
                \label{pipeline_schedulingsub-problem}
	\end{align}
    Thus, the subproblem $\mathcal{P}3$ is an integer linear programming (ILP) problem. We have the following theorem.

\begin{theorem}
    For Problem $\mathcal{P}3$,  the optimal micro-batch size $b^*$ is given by
\begin{equation}
b^*=\left\{ \begin{array}{l}
b_1^*, \ 0 < b\leqslant \min\{b_{\mathrm{th}}^c,b_{\mathrm{th}}^s\},\\
	b_2^*, \ 0<\max\{b_{\mathrm{th}}^s,b_{\mathrm{th}}^c\} \leqslant b,\\
	b_3^*, \ 0<b_{\mathrm{th}}^c<b < b_{\mathrm{th}}^s,\\
	b_4^*, \ 0<  b_{\mathrm{th}}^s<  b < b_{\mathrm{th}}^c,\\
\end{array} \right. 
\label{optimal_batch}
\end{equation}
Moreover, $b_1^*$, $b_2^*$, $b_3^*$ and $b_4^*$ are derived in \eqref{b_1}, \eqref{b_2}, \eqref{b_3} and \eqref{b_4}, respectively.
  \label{proposition_for_pipe}
\end{theorem}
\begin{proof}
Please see Appendix A.
\end{proof}

\subsection{Solution to the MSP Problem}

\begin{figure}[t!]
\centering
\includegraphics[width=8cm]{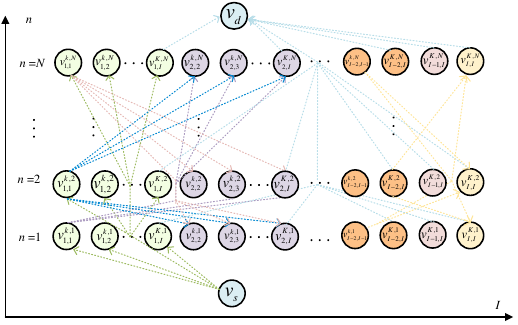}
\caption{The graph representation of the MSP problem.}
\label{original-graph}
\end{figure}

By fixing the micro-batch $b$, the MSP problem can be formulated as
	\begin{align}
	\mathcal{P}4:	&\underset{\{ \boldsymbol{x}, \boldsymbol{y}, T_1 \}}{\mathrm{min}} \quad T_f(\boldsymbol{x}, \boldsymbol{y}) +\xi  \cdot T_1\notag \\
          \mathrm{s.t.}\quad
          & \mathrm{C1},\ \mathrm{C2}, \ \mathrm{C4}\text{ - } \mathrm{C16},
    \label{problem_mapping}
	\end{align}
which is a Mixed Integer Nonlinear Programming (MINLP) problem. Fortunately, we find that it falls into a combinatorial optimization problem (COP) with a combined min-max and min-sum objective function \cite{minoux1989solving}. To tackle this problem, we convert the search space for the split and placement decisions into an undirected graph $\mathbf{G}=\left( \mathbf{V},\mathbf{E} \right) $, where the set of vertices can be expressed by
\begin{equation}
\begin{aligned}
    \mathbf{V}=\{ v_{(i-m),i}^{k,n}|\ \forall &k \in  [1,K],\ n \in  \mathcal{N},\\
    & \ i \in [1,I],\ m \in [0, i-1] \},
    \end{aligned}
\end{equation}
where a single vertex $v_{(i-m),i}^{k,n}$ represents the decision of allocating the $(i-m)$-th layer to the $i$-th layer to the $k$-th submodel on the $n$-th edge node with all clients grouped into one virtual node for $k=1$ and servers for $k\in[2,K]$. Here, $m \in [0,i-1]$. Besides, the set of vertices $\mathbf{V}$ comprise all possible assignment decisions. Then, we can easily connect the vertex following Constraints $\mathrm{C1},\ \mathrm{C2} $ and $\mathrm{C4}\text{-}\mathrm{C16}$
\begin{equation}
\begin{aligned}
\mathbf{E}=\big\{ &\big( v_{(i-m),i}^{k,n},v_{(i+1),(i+1+m^\prime)}^{(k+1),n^\prime} \big) |\ \forall k \in [1,K-1],\\
&\ n, \ n^\prime \in \mathcal{N}, \ n \neq n^\prime,  \ i \in [1,I],\\
&\ m \in [0, i-1],\ m^\prime \in [0, I-i-1]\big\}, 
    \end{aligned}
\end{equation}
where an edge $\big( v_{(i-m),i}^{k,n},v_{(i+1),(i+1+m^\prime)}^{(k+1),n^\prime} \big)$ represents choosing assignment $v_{(i+1),(i+1+m^\prime)}^{k^\prime,n^\prime} $ after assignment $v_{(i-m),i}^{k,n}$. 
We set the weight of the edge to be the sum of latency of transferring data between the submodel $k$ and $(k+1)$ and the computing latency on client/server $n$, which can be expressed as
\begin{equation}
\begin{aligned}
&c\left( v_{(i-m),i}^{k,n},v_{(i+1),(i+1+m^{\prime})}^{(k+1),n^{\prime}} \right) =
\\
&\left\{ \begin{array}{l}
	\max\limits_{n\in \mathcal{N}_c } \{t_{1,n}^{\mathrm{F}}(b,\boldsymbol{x})+\sum\limits_{n^{\prime}=1}^N{y_{2n^{\prime}}}t_{1,nn^{\prime}}^{\mathrm{F}}(\boldsymbol{x},b)\}+\max\limits_{n\in \mathcal{N}_c } \{\\
	\quad \quad  \quad    \quad \quad      \    t_{1,n}^{\mathrm{B}}(b,\boldsymbol{x})+\sum\limits_{n^{\prime}=1}^N{y_{2n^{\prime}}}t_{1,n^{\prime}n}^{\mathrm{B}}(\boldsymbol{x},b)\},\ k=1,\\
	t_{k,nn^{\prime}}^{\mathrm{F}}(\boldsymbol{x},b)+t_{k,n^{\prime}n}^{\mathrm{B}}(\boldsymbol{x},b)+t_{k,n}^{\mathrm{F}}(b,\boldsymbol{x})+t_{k,n}^{\mathrm{B}}(b,\boldsymbol{x}),\\
   \quad \quad \quad  \quad \quad \quad  \quad \quad \quad  \quad \quad \quad \quad \quad  \quad \quad \
 \ \forall k\in [2,K-1],\\
\end{array} \right. 
    \end{aligned}
\end{equation}
where computing latency is set to be the maximum computing latency of all clients if $k=1$.
Furthermore, we add the virtual source vertex $v_s$ and destination vertex $v_d$, respectively, where the weights of edges between $v_s$ or $v_d$ and any edge node in the graph are set to zero, i.e., $c\left( v_s,v_{1,(1+m^\prime)}^{1,n^\prime} \right)=0$ and $c\left( v_{(I-m),I}^{K,n},v_d \right)=0$.

To solve the MSP problem with the combined min-max and min-sum objective, we propose an efficient algorithm to find the optimal solution via the defined graph. 
Then, we sort the cost of edges in descending order to construct subgraphs. Given an edge $e\in \mathbf{E}$, we generate a subgraph that each path from the source vertex to the destination vertex in the subgraph should contain $e$. The process is repeated until we go through all the edges.
With each subgraph, we first use reformulation-linearization technique
to handle the original min-sum part in Problem $\mathcal{P}4$ which will be elaborated upon in Appendix B  and calculate a lower bound for the objective to decide whether to search the subgraph or discard the graph.
Once searching, the original MSP problem in $\mathcal{P}4$ becomes finding the shortest path from $v_s$ to $v_d$, and the vertices along this path form the layer splitting and model placement decisions $\{x_{ik},\ y_{kn}\}$, i.e., $x_{ik}=1$ and $y_{kn}=1$ if and only if $v_{(i-m),i}^{k,n}$ is on the path. By transforming the COP Problem $\mathcal{P}4$ into a graph, the heap-optimized Dijkstra's algorithm~\cite{haeupler2024universaloptimalitydijkstrabeyondworstcase} is used to find the shortest path within each generated subgraph. The minimum bound of the total latency can be found by comparing the  min-sum-min-max objectives by assuming $e$ is the bottleneck link  across all the subgraphs. For this reason, we call the process as bottleneck-aware shortest path algorithm, and the detailed procedure is summarized in \textbf{Algorithm \ref{alg_1}}. We have the following results.

\begin{algorithm}[t] 
	\caption{Bottleneck-aware Shortest Path Algorithm for the MSP Problem.}
 \label{alg_1}
	\LinesNumbered 
	\KwIn{Graph $G = (V, E)$, the maximum GPU memory of the client or server $M_n$, the number of edge servers $N$, the neural network to be processed, micro-batch size $b$,  mini-batch size $B$, channel-related parameters $W_{nn^\prime}$, $N_0$,  $r_{nn^{\prime}}$ and $\gamma$, server-related parameters  $p_n$, $t_0^s$, $t_0^c$, $t_1^s$, $t_1^c$, $f_n$,  $\kappa_n$,  ${L_{t}}^*=+\infty $.
}
	\KwOut{The optimal training latency $L_t(\boldsymbol{x}, \boldsymbol{y}, T_1)$ per round, splitting decisions $\{x^*_{ik}|i\in[1,I],k\in[1,K]\}$, and placement decisions $\{y^*_{kn}|k\in[1,K],n\in[1,N] \}$.}
     Calculate the weights $c\left( v_{(i-m),i}^{k,n},v_{(i+1),(i+1+m^\prime)}^{(k+1),n^\prime} \right)$ of all edges\;
 Sort edges $e$ in descending order of their weight $w(e)$, resulting in set $\mathbf{E}_w$\;
   Find the lower bound $l_b$ for the orginal graph with LP solver by reformulating and linearizing the min-sum part in Problem $\mathcal{P}4$\;
 \For{ \textnormal{$e \in \mathbf{E}_w$}}{
  Generates a subgraph from $G$ containing edge $e$ where the cost of $e$ is the maximum one in the subgraph\;
        \If{ $l_b+\xi \cdot w(e)> {L_{t}}^*$}{
    Continue; 
}
  Find an optimal path in the subgraph that includes edge $e$ using the heap-optimized Dijkstra's algorithm satisfying the min-sum function $\varpi_h = {\mathrm{min}}\ T_f(\boldsymbol{x}, \boldsymbol{y}) $ and the memory constraint\;
  Compute the min-sum-min-max objective  ${L_{t}}^h=\varpi_h+ \xi \cdot w(e)$\;
   \If{${L_{t}}^h <{L_{t}}^*$}{Update the training latency ${L_{t}}^* \gets {L_{t}}^h$\;
   $x^*_{ik} \gets x_{ik} \ \text{and}\ y^*_{kn} \gets y_{kn}$\;}
   Exclude $e$ from Graph $G$;
   }
\end{algorithm}

\textit{\textbf{Finding lower bounds and pruning:}} 
To efficiently reduce the search space of bottleneck-aware shortest path algorithm for the MSP problem, we introduce the reformulation-linearization approach \cite{FRIEZE198389} to compute a lower bound for the objective (see Appendix B). 
Specifically, the bound is achieved by relaxing quadratic constraints of the original problem. 
The relaxed problem is then solved as an  linear programming (LP) problem using LP solver Gurobi \cite{gurobi},
which computes the lower bound for min-sum function in the MSP problem.
When repeating searching subgraphs, if the lower bound of min-sum part plus the cost of current edge   exceeds the current best solution ${L_{t}}^*$, the subsequent subgraphs  can be safely excluded from further exploration.
In other words, subproblems that cannot achieve a better solution than the current best solution 
${L_{t}}^*$
  can be terminated early, saving computation time
 (see \textbf{Algorithm \ref{alg_1}}).

\begin{theorem}
\textbf{Algorithm \ref{alg_1}} can obtain the optimal solution to the MSP subproblem $\mathcal{P}4$.
\end{theorem}
\begin{proof} 
The MSP Problem $\mathcal{P}4$ is a min-sum-min-max problem, which can be optimally solved by our bottleneck-aware shortest path algorithm. The proof is omitted due to page limitation, and a similar proof can be found in \cite{minoux1989solving}.
\end{proof}

\begin{theorem}
\textbf{Algorithm \ref{alg_1}} has the computational complexity  is $O(E\log E + E(V+E)\log V) = O(E(V+E)\log V)$.
\end{theorem}
\begin{proof} 
 We sort the edges of a graph according to their respective costs with the complexity of $O(E\log E)$. By checking every edge to generate a subgraph containing the current edge, the complexity is $O(E)$. Then, the total complexity of subgraph generation is $O(E\log E)$.
Running the heap-optimized Dijkstra’s algorithm on each subgraph has a complexity of $O((V^\prime+E^\prime)\log V^\prime+(V^{\prime \prime}+E^{\prime \prime} )\log V^{\prime\prime})$,
where $V^\prime$ and $E^\prime$ are the number of vertices and edges in the subgraph before edge $e$, and $V^{\prime \prime}$ and $E^{\prime \prime}$ are the number of vertices and edges in the subgraph after edge $e$.
In the worst case, the subgraph size is similar to the original graph, which can be approximated as $O((V+E)\log V)$.
Thus, the total complexity of the approach is
$O(E\log E + E(V+E)\log V) = O(E(V+E)\log V)$.
\end{proof}

 \begin{algorithm}[t] 
	\caption{BCD-based Splitting, Placement, and Micro-batching Algorithm.}
   \label{total-alg}
	\LinesNumbered 
	\KwIn{Convergence tolerance $\vartheta$, iteration index $\tau=0$, maximum GPU memory of the $n$-th server $M_n^s$.
}
	\KwOut{$b^*$, $\bm X^*$, $\bm Y^*$.}
 Initialization:  $b^{0}$\;
		\While{$|L_t(x_{ik}^\tau, y_{kn}^\tau, b^\tau) - L_t(x_{ik}^{\tau-1}, y_{kn}^{\tau-1}, b^{\tau-1})| \geq \vartheta $}{
            $ \tau \gets \tau +1$\;
   			Obtain $x_{ik}^{\tau}$, $y_{kn}^{\tau}$ and $T_1^{\tau}$ by \textbf{Algorithm \ref{alg_1}} with fixed decision variables $b^{\tau-1}$\;
      Obtain $b^\tau$ based on (\ref{optimal_batch}) with fixed decision variables $x_{ik}^{\tau}$, $y_{kn}^{\tau}$ and $T_1^{\tau}$\;
$b^* = b^\tau$, $\bm X^* = \{x_{ik}^{\tau}\}$, and $\bm Y^*=\{y_{kn}^{\tau}\}$.   
		}
  \label{alg_2}
\end{algorithm}

\subsection{Splitting, Placement, and Micro-batching Design}
As mentioned, we decompose the original MINLP Problem $\mathcal{P}2$ into an MSP subproblem and a micro-batching subproblem (i.e., $\mathcal{P}3$ and $\mathcal{P}4$), and find the optimal solution to each subproblem. Then, we further propose a block-coordinate descent (BCD)-based algorithm~\cite{tseng2001convergence} to solve the original Problem $\mathcal{P}2$, which is detailed in~\textbf{Algorithm~\ref{total-alg}}, where $x_{ik}^{\tau}$, $y_{kn}^{\tau}$, $T_1^{\tau}$ and  $b^\tau$ represent $x_{ik}$, $y_{kn}$, $T_1$,  and $b$  at the $\tau$-th iteration.
Despite the mixed-integer nature~\cite{grippo2000convergence} of the original Problem $\mathcal{P}2$, the BCD-based algorithm converges within only a few iterations. Note that by optimally solving the original variables in each iteration,
the convergence of the BCD procedure can be guaranteed because it finds the optimal solution at each step.
Our simulations in Section \ref{simulations} will further demonstrate that the proposed algorithm can achieve the near-optimal solution to the joint optimization problem.

\section{Simulations}\label{simulations}
\subsection{Simulation Setup}
\textcolor{blue}{
\begin{table}[t]\label{table_3}
  \centering
  \caption{Simulation Parameters.}
  \renewcommand{\arraystretch}{1.25}{
  \setlength{\tabcolsep}{0.5mm}{
\begin{tabular}{|c|c|c|c|}
\hline
\textbf{Parameter}          & \textbf{Value} & \textbf{Parameter} & \textbf{Value}  \\ \hline
$f_n$             & $1\sim10  $ TFLOPS              & $I$                 & $16 $                         \\ \hline
$B$             & 512              & $M_n$          & $2\sim16$ GB                   \\ \hline
$W_{nn^{\prime}}$               & $10 \sim 200$ MHz           & $\vartheta$                  & 0.01                       \\ \hline
$N_0$        & $-174$ dBm/Hz            & $b^0$             & 20                       \\ \hline
$\kappa _{n}$        & $\frac{1}{{32}}$ FLOPs/byte            & $N$             & $2\sim10$                     \\ \hline
$d_{nn^{\prime}}$        & $1\sim500$ m            &$t_0^c$/$t_0^s$            &$0.001$ s                        \\ \hline
$\gamma $       & $3.5 $            &$t_1^c$/$t_1^s$            &$ 0.001 $ s                       \\ \hline
$p_n$        & $100\sim500$ mW            &[$b_{\text{th}}^c$, $b_{\text{th}}^s$]            &[$32$, $32$]                         \\ \hline
\end{tabular}}}
\label{parameter_set}
\vspace{-0.8em}
\end{table}
}

In our simulations, $N$ edge servers are deployed within a square area of 0.5 km $\times$ 0.5 km where the distance between any two servers is within 500 m.
By default, $N$ is set to 6. The computing capability $f_n$ is  distributed within $[1, 10]$ TFLOPS \cite{10415235}. We consider two cases of 5G Integrated Access Backhaul (IAB) scenarios to evaluate high-speed and low-speed communication networks. For the low-speed case, we consider 5G sub-6GHz by setting the available bandwidth per link from 10 MHz to 50 MHz for access and backhaul links; For high-speed case, we consider 5G mmWave with available bandwidth per link ranging from 100 MHz to 200 MHz. Besides, the noise spectral density of servers is set to $-174$ dBm/Hz~\cite{hu2019edge}. Moreover, we consider the signal power $p_n$ from 100 to 500 mW, and the computing intensity ${\kappa _n} ={1}/{{32}}$ FLOPs/byte. The path loss exponent is set to 3.5 according to the channel model in~\cite{samimi2015probabilistic}. For readers' convenience, the detailed simulation parameters are summarized in Table \ref{parameter_set}.
We adopt the image classification datasets MNIST and CIFAR-10~\cite{lecun1998mnist} to evaluate the performance of the proposed pipelined SL for the model VGG-16~\cite{simonyan2014very} under both IID and non-IID settings. 
 Moreover, the mini-batch size $B$ is set to 512, with the initial micro batch $b$ set to 20 and [$b_{\text{th}}^c$, $b_{\text{th}}^s$] set to [32, 32]. Accordingly, based on the experiments in \cite{9252924}, $t_0^c$, $t_0^s$, $ t_1^c$ and $t_1^s$ are set to 0.001 s to match $b_{\text{th}}^c$ and  $b_{\text{th}}^s$.
 The algorithms have been implemented on a computer equipped with an AMD Ryzen Threadripper PRO 5975WX and NVIDIA GeForce RTX 4090.

\begin{figure}[t!]
\centering\subfigure[\centering CIFAR-10 with IID settings.] {
		\centering\includegraphics[width=3.85cm]{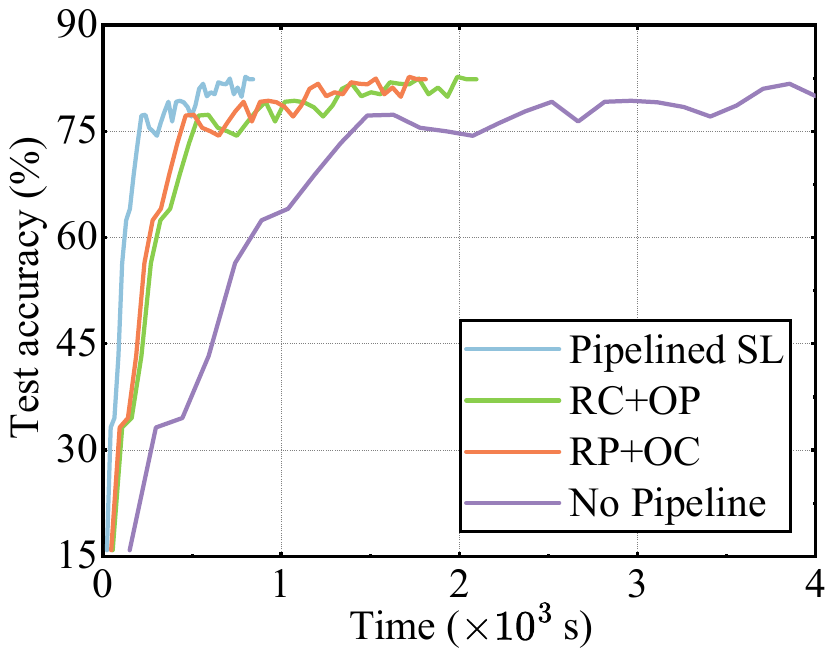}	
	} 
         \subfigure[\centering CIFAR-10 with non-IID settings.] {
		\centering\includegraphics[width=3.85cm]{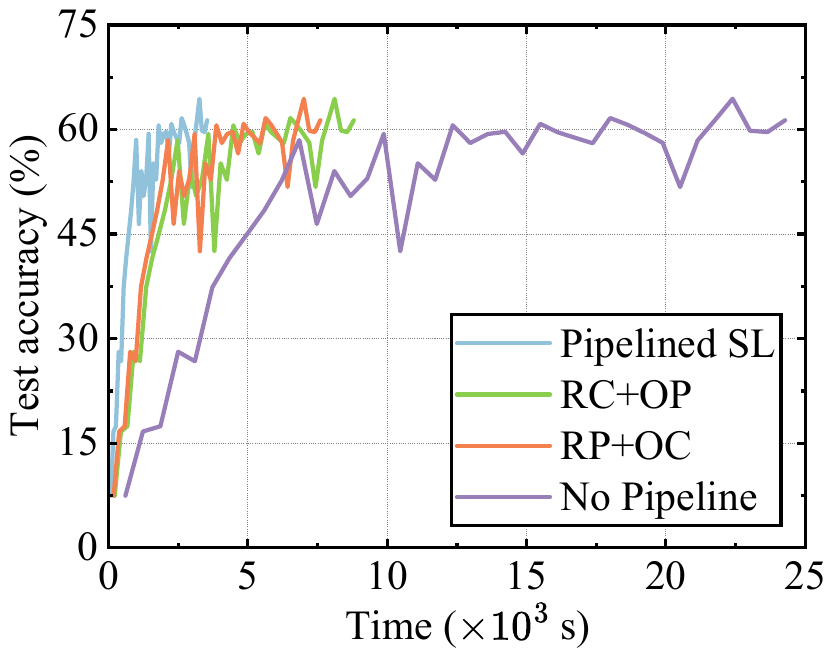}
	}
 
 \centering\subfigure[\centering MNIST with IID settings.] {
		\centering\includegraphics[width=3.85cm]{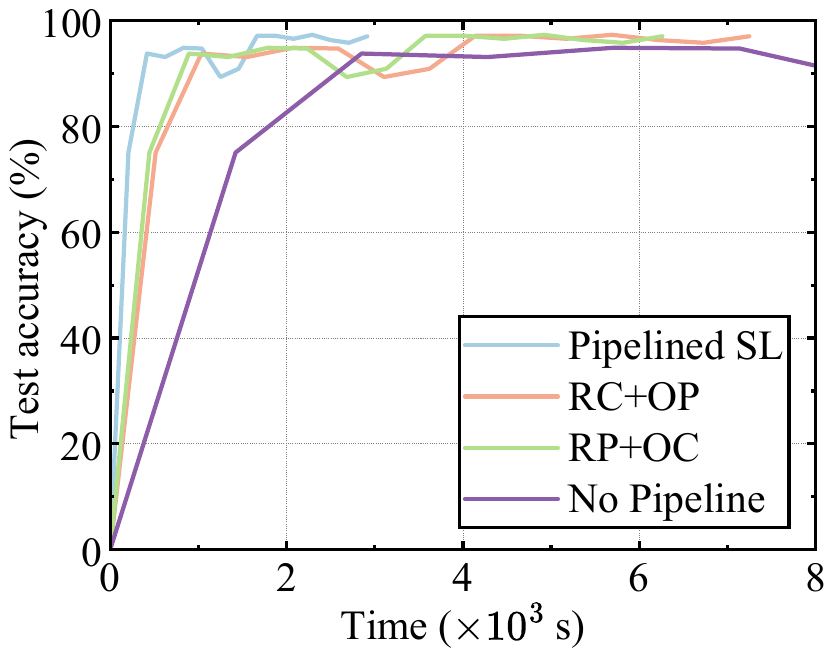}
	} 
         \subfigure[\centering MNIST with non-IID settings.] {
 \centering
		\centering\includegraphics[width=3.85cm]{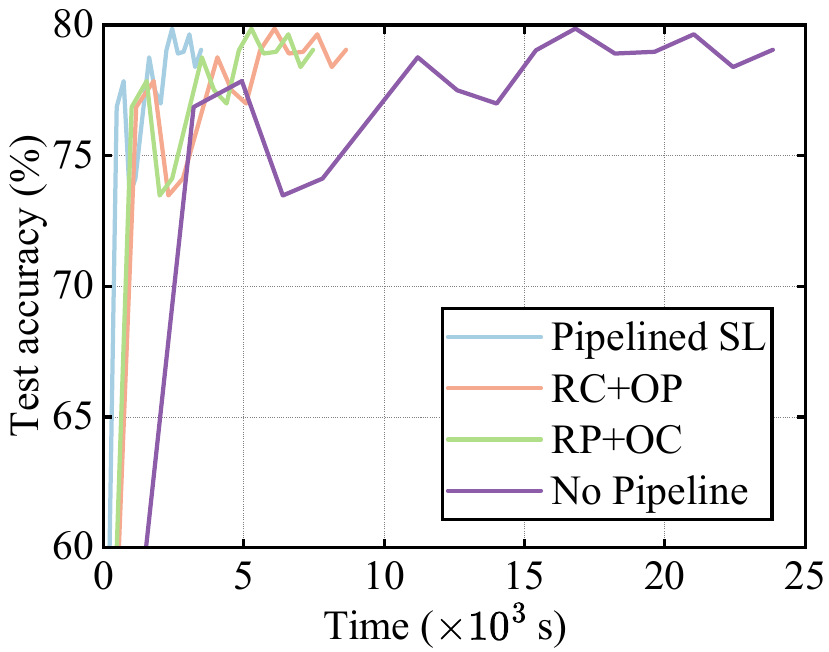}
        }
	\centering\caption{The test accuracy of different SL schemes on CIFAR-10 and MNIST datasets using VGG-16.} 
		\label{Training_curves}
		\vspace{-0.5em}
\end{figure}

To demonstrate the advantages of our framework, we compare it with three benchmarks: 1) ``Random Cut and Optimal Placement'' (\textbf{RC + OP}) randomly partition the model into several submodels while adopting the proposed placement strategy. 2) ``Random Placement and Optimal Cut'' (\textbf{RP + OC}) employs the proposed model splitting strategy while randomly placing the submodels among edge servers. 3) \textbf{No Pipeline} employs our algorithm to optimize both model splitting and placement without dividing a mini-batch into multiple micro-batches for pipelining. Due to the optimality, ``No Pipeline'' also serves as the performance upper bound of directly applying the existing split inference/learning scheme without pipeline parallelism~\cite{9562523}. In many experiments, we defined total latency as the time required for the model to reach 75\% test accuracy on CIFAR-10 under  IID setting.

\subsection{Performance Evaluation of the Proposed Pipelined SL Framework}
Fig. \ref{Training_curves} evaluates the test accuracy of our proposed pipelined SL scheme and other benchmarks across CIFAR-10 and MNIST datasets under both IID and non-IID settings. As observed, our scheme outperforms other benchmarks, including ``RC + OP'' and ``RP + OC'' schemes, demonstrating that both splitting and placement play a crucial role in enhancing training efficiency. More importantly, all the schemes substantially outperform the ``No Pipeline'' scheme since the ``No Pipeline" approach requires almost $3\times$ to $7\times$ training time to reach the same accuracy level compared with other schemes. This underscores the effectiveness of pipelining in SL for multi-hop edge networks. By enabling parallel processing across multiple edge nodes, pipelined SL efficiently utilizes available resources, leading to faster model convergence without sacrificing accuracy.

\begin{figure}[t!]
	\centering\subfigure[\centering Total latency v.s. number of servers.]{ 
	{
			\centering
\includegraphics[width=3.85cm]{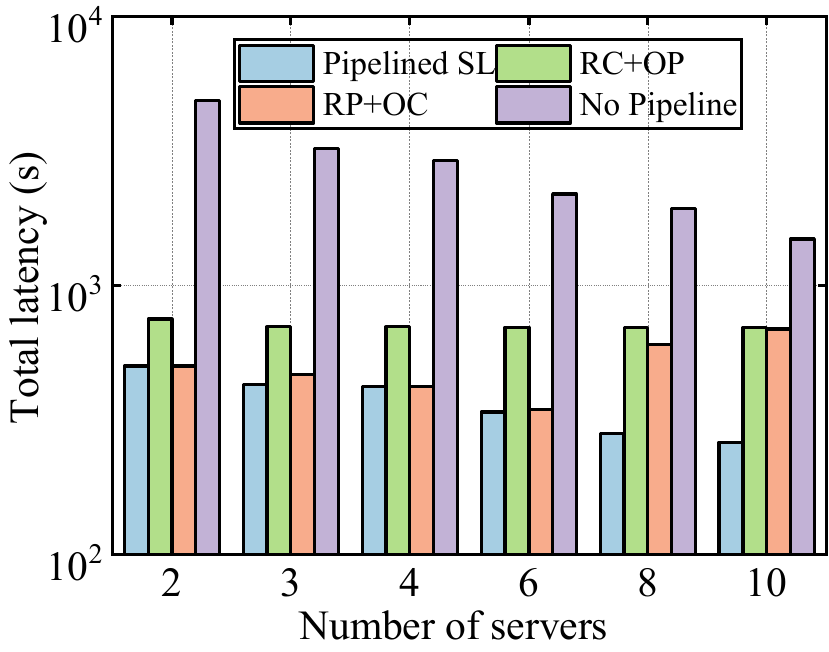}
		
	}}
	\centering\subfigure[\centering Total latency v.s. bandwidth.]{ 
	{
			\centering
			\includegraphics[width=3.85cm]{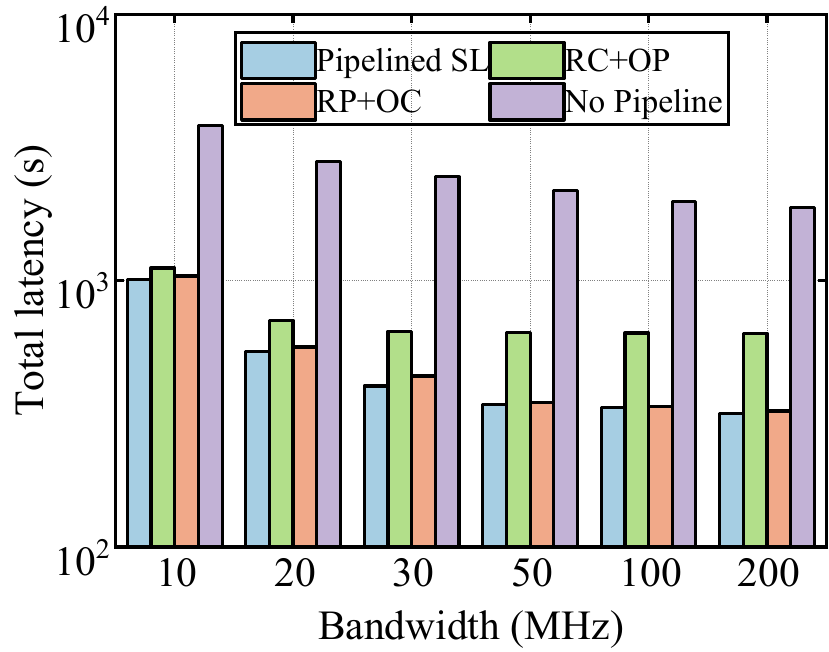}
	}}
 
        \centering\subfigure[\centering Total latency v.s. computing ability.]{ 
	{
			\centering
			\includegraphics[width=3.85cm]{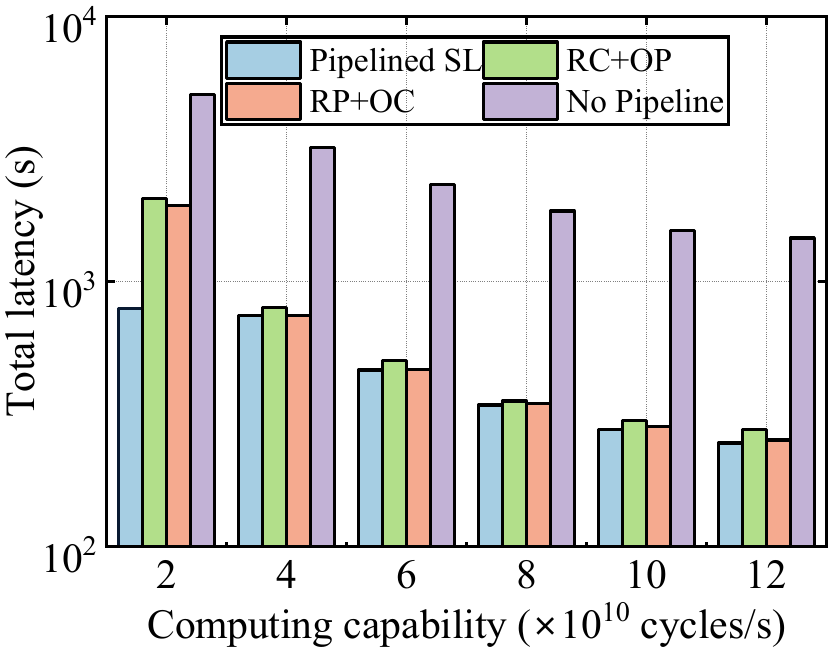}
	}}
	\centering\subfigure[\centering Total latency v.s. memory (B=4096).] {
	{
		
			\centering
			\includegraphics[width=3.85cm]{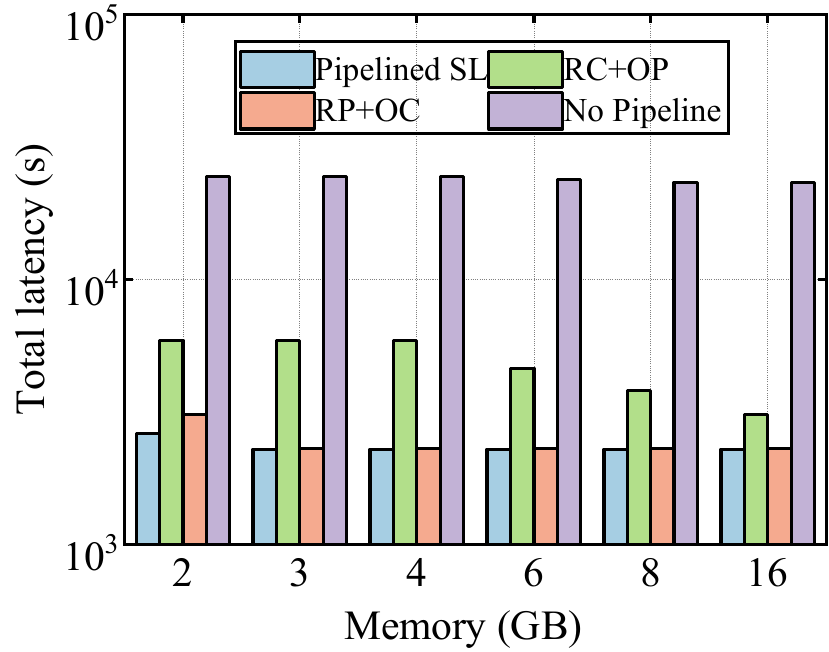}
		
	}}
	\centering\caption{The total latency versus varied network settings.}
		\label{Round_latency}
		\vspace{-0.5em}
\end{figure}

Notice that different schemes indeed have the same converged accuracy because these schemes are essentially equivalent to centralized training, except that their total latency is different. In Fig. \ref{Round_latency}, we vary the networking settings to illustrate the total latency among the aforementioned schemes. 
As shown in Fig.~\ref{Round_latency}(a), when we vary the number of servers from 2 to 10, the total latency for all schemes decreases significantly.  This reduction is due to increased parallelism and more efficient model splitting and placement across more servers. The edge network benefits from enhanced flexibility, allowing for better utilization of computational resources.
In Fig.~\ref{Round_latency}(b), we observe the impact of bandwidth on total latency ranging from 10~MHz to 200~MHz. The total latency of all schemes decreases as bandwidth increases. This demonstrates that higher communication bandwidth reduces data transmission time between edge servers, which is crucial for SL where frequent inter-server communication occurs.
Fig.~\ref{Round_latency}(c) illustrates the impact of computing capabilities on latency by varying computing capability from $2 \times 10^{10}$  to $12 \times 10^{10}$ cycles/s.  It is unsurprising to see that enhanced computing power at edge servers shortens the local processing time required for training.
Finally, Fig.~\ref{Round_latency}(d) shows the effect of memory capacity on total latency by varying it from 2~GB to 16~GB. The latency for all schemes decreases with increased memory capacity, because larger memory space allows each server to handle more layers or larger micro-batches, adding flexibility in model splitting and placement.

In edge networks, computing and communication resources often fluctuate during training, resulting in discrepancies between the measured conditions (used for optimization) and the actual network conditions. To assess the impact of these measurement errors, we model the fluctuations in data rates and computing capabilities by introducing Gaussian noise with varying coefficients of variation (CV)~\cite{10053757,yoo2024modeling}. As illustrated in Fig.~\ref{fluctuation}, our scheme maintains robustness across different levels of variation, resulting in only minor changes in overall latency. These results validate the effectiveness of our pipelined SL approach in fluctuating edge environments.

\begin{figure}[t!]
\centering\subfigure[\centering Computing capability uncertainty.] 
	{
	
			\centering\includegraphics[width=3.85cm]{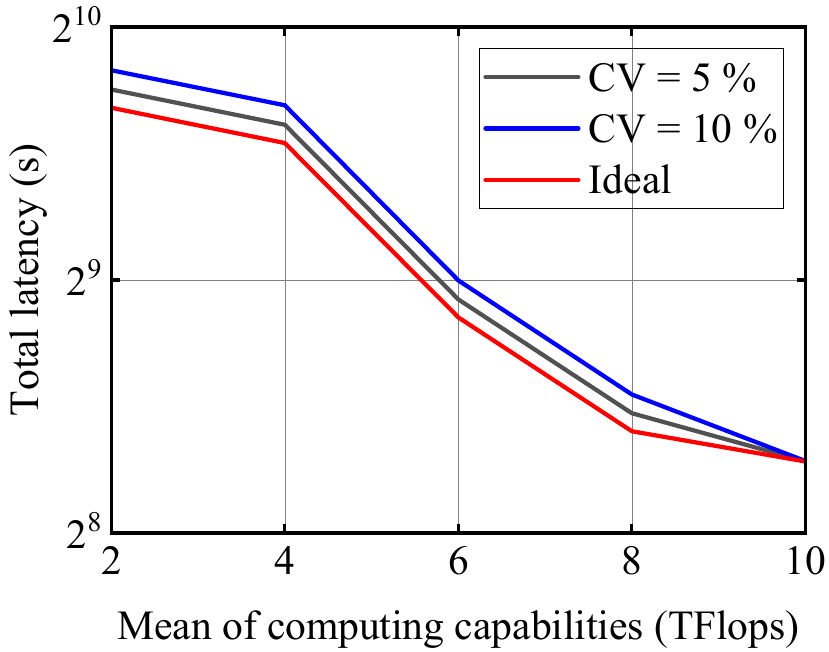}
	}
\subfigure[\centering Uplink transmission rate uncertainty.] 
	{
			\centering\includegraphics[width=3.85cm]{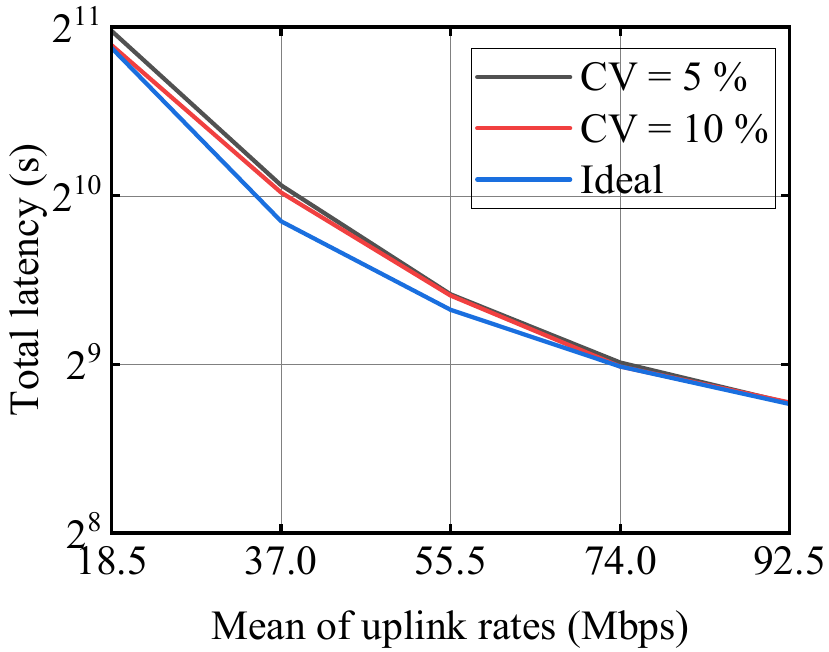}
	}
	\centering\caption{\centering The impact of network resource fluctuation on total latency on CIFAR-10 under IID setting using VGG-16.}
		\label{fluctuation}
		\vspace{-0.5em}
\end{figure}

\begin{figure}[t!]
\centering\subfigure[\centering Total latency.] 
	{
	
			\centering\includegraphics[width=3.85cm]{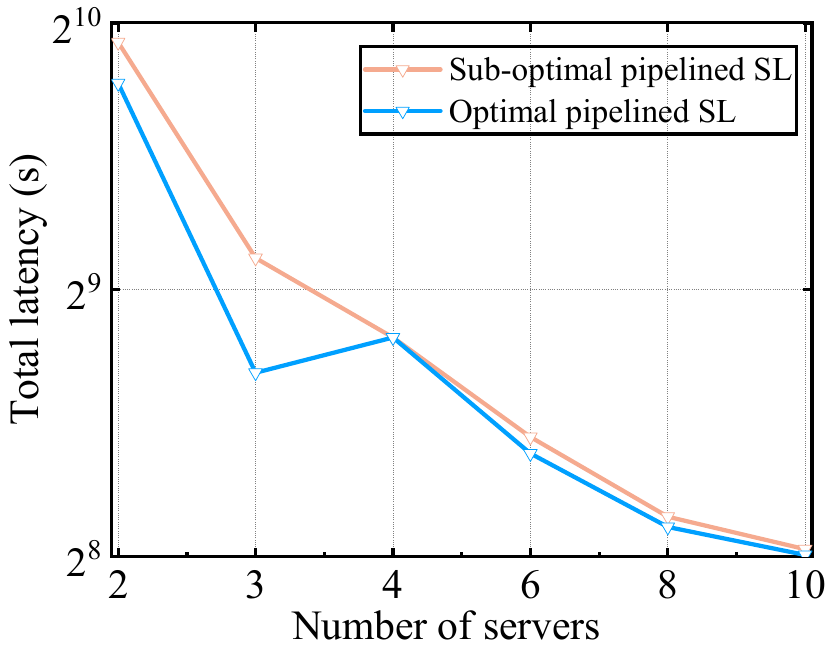}
		
	}
\centering\subfigure[Running time.] 
	{
		
			\centering\includegraphics[width=3.85cm]{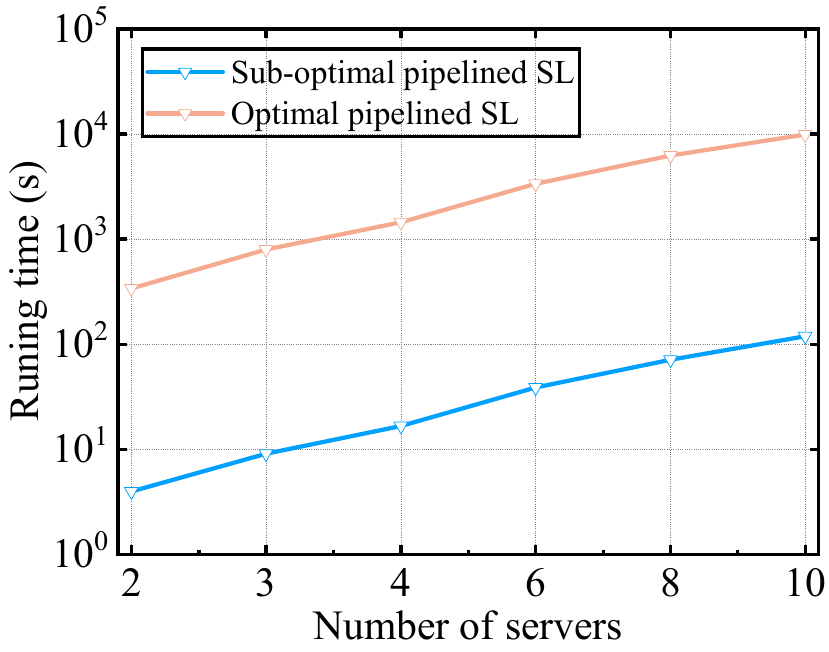}
	
	}
	\centering\caption{ The performance comparison of the sub-optimal and optimal pipelined SL.}
		\label{Compare_optimal}
		\vspace{-0.5em}
\end{figure}

\begin{figure}[t!]
\centering\subfigure[\centering Total latency v.s. network topology (N=8).] 
	{
	
			\centering\includegraphics[width=3.85cm]{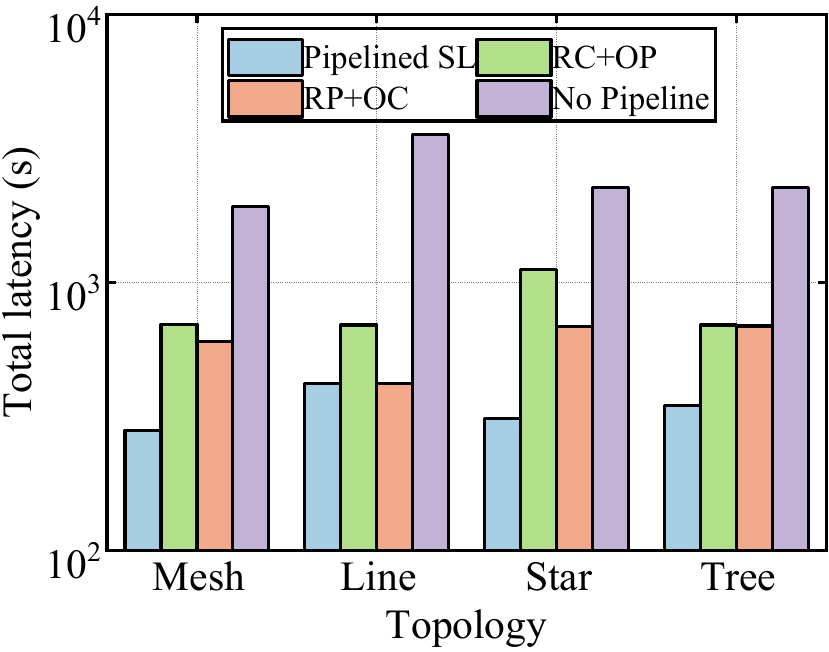}
		
	}
\centering\subfigure[\centering Total latency v.s. number of servers.] 
	{
		
			\centering\includegraphics[width=3.85cm]{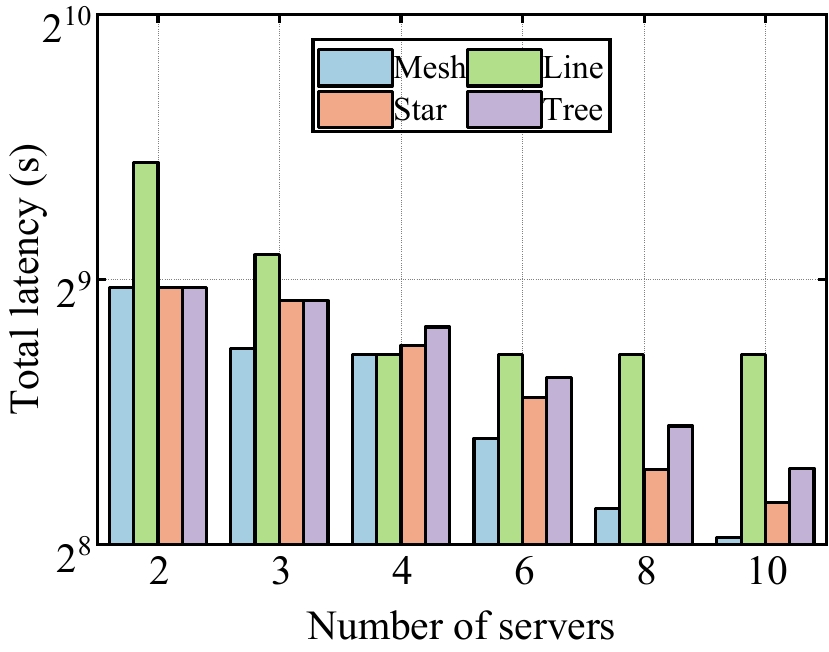}
	
	}
	\centering\caption{ Performance evaluation of the proposed pipelined SL framework across varied network topology.}
		\label{Compare_topology}
		\vspace{-0.5em}
\end{figure}

In Fig. \ref{Compare_optimal}, we evaluate the performance of our proposed suboptimal BCD algorithm and the optimal scheme, where the optimal scheme performs an exhaustive search on every possible micro-batch size, based on which we run Algorithm \ref{alg_1} to find the MSP solutions for each micro-batch size.
Specifically, 
in Fig. \ref{Compare_optimal}(a), we observe that the total latency of our suboptimal solution closely matches that of the optimal scheme. For instance, with 10 servers, the optimal scheme achieves a total latency of approximately $2.569 \times 10^{2}$ seconds, while the suboptimal BCD   algorithm attains approximately 
$2.609\times 10^{2}$ seconds. This confirms that suboptimal BCD algorithm effectively provides near-optimal performance.
Furthermore, Fig. \ref{Compare_optimal}(b) shows that as the number of servers increases from 2 to 10, the running time of the optimal scheme escalates drastically, reaching up to $9.884\times 10^3$ seconds (over 2.745 hours) at 10 servers. In contrast, the suboptimal BCD algorithm has a maximum running time of 119.132 seconds, demonstrating its scalability and efficiency. In practice, the running time is feasible even for the 10-server scenario since we consider a stationary edge network and the total training time can be for hours. In summary, our proposed algorithm achieves a suboptimal solution close to the optimal one while taking much less time to find the solution.

In Fig. \ref{Compare_topology},
 to validate the effectiveness of our solotion
to the MSP problem and the micro-batching problem across
different physical network topologies, 
we compare the total latency for varied network topologies, including mesh, line, star, and tree structures.
Specifically, Fig. \ref{Compare_topology}(a) illustrates the relationship between total latency and varied network topologies.  
 Mesh topology shows low latency due to high connectivity, while line topology has higher latency because of sequential connections. 
 In the star topology, the central node needs to forward messages between peripheral nodes, resulting in increased total latency due to the forwarding overhead at the central node.
 Tree topologies exhibit higher total latencies, attributed to their hierarchical communication constraints. 
In Fig. \ref{Compare_topology}(b),  We notice when 
$N=2$, the mesh, star, and tree topologies exhibit identical total latency, while the unidirectional line topology experiences the longest total latency.  Moreover, when we vary the number of servers from 2 to 10, the total latency for all network topologies decreases significantly.  This reduction is due to increased parallelism and more efficient model splitting and placement across more servers.



\section{Conclusion}\label{conclusion}
In this paper, we addressed the joint optimization problem of model splitting and placement (MSP) and micro-batching for pipelined split learning (SL) in multi-hop edge networks under resource constraints. By taking pipeline parallelism into account, we discovered that the resulting MSP problem has a combined min-max and min-sum objective. Based on graph theory, we devise a bottleneck-aware shortest-path algorithm to solve the MSP problem optimally. Moreover, with fixed MSP outcomes, we also obtain the optimal micro-batch size for pipelined SL under communication-computing resource constraints. To minimize end-to-end training latency, we introduce an iterative algorithm that effectively solves the joint optimization problem. Our simulation findings consistently demonstrate the superior performance of our proposed approach compared with SL without pipeline parallelism.
\appendices
\section{Proof of Theorem \ref{proposition_for_pipe}}
To optimize the micro-batching subproblem $\mathcal{P}3$ defined in (\ref{pipeline_schedulingsub-problem}), 
we consider the piecewise nature of the function induced by $b_{\mathrm{th}}^s$ and $b_{\mathrm{th}}^c$ and analyze each case separately. 
First, we arrange 
$
\left\lfloor \frac{b}{M} \right\rfloor
$
data samples to each of the first \(M-1\) clients and allocate the remaining
$
b - (M-1)\left\lfloor \frac{b}{M} \right\rfloor
$
data samples to the \(M\)-th client.
Thus, we have several situations:
\begin{figure*}[ht]
\begin{equation}
\begin{aligned}
    &\mathcal{Z}_1 =  b \bigg\{ {\sum_{k=2}^K{\sum_{n=1}^{N}{\frac{y_{kn}\kappa_n\delta _{kn}^{\mathrm{F}}}{f_n}}}}
+\sum\limits_{n=1}^{N}\sum_{n^{\prime}=1}^{N}\sum_{i=1}^{I} {\frac{\sum\limits_{k=2}^{K-1}{y_{kn}}y_{(k+1) n^{\prime}}  x_{ik}\varphi_i+\sum\limits_{k=2}^{K-1}{y_{kn}}y_{(k+1)n^{\prime}}x_{ik}\phi_{(i+1)}}{W_{nn^{\prime}}\cdot \log \left( 1+{p_n{(d_{nn^{\prime}})}^{-\gamma}}/{N_0} \right)}}
         \bigg\} +T_1 \bigg\lceil \frac{B-b}{b}\bigg \rceil  + (K-1)(t_0^s\\
          & +t_1^s) +\max_{n \in \mathcal{N}_c}\{b_m\kappa _n\delta_{1n}^{\mathrm{F}}
 /f_n+t_0^c + \sum\limits_{n^\prime=1}^Ny_{2n^\prime}b_m  \sum\limits_{i=1}^{I}x_{i1}\varphi_i/r_{nn^\prime}\}+\max_{n \in \mathcal{N}_c}\{t_1^c + \sum\limits_{n^\prime=1}^Ny_{2n^\prime}b_m  \sum\limits_{i=1}^{I}x_{i1}\phi_{(i+1)}/r_{nn^\prime}\}
          .
  \end{aligned}
   \label{zb1}
\end{equation}
 \end{figure*}
   \begin{figure*}[ht]
   \vspace{-0.8em}
  \begin{align}
{b_{v}}= 
\min \bigg\{& \min\limits_{n\in \mathcal{N}\setminus \mathcal{N}_c} \bigg\{ 
\bigg\lfloor \frac{M_n}{\big\{x_{i1}+\sum\limits_{k=2}^K{y_{kn} (x_{ik}-x_{i(k-1)})\big\}\sum\limits_{i=1}^{I}(
 {\widetilde \varphi} _i
 +\widetilde{\phi}_i + {\widetilde \sigma}_i + \beta_i)}} \bigg\rfloor, 
 \bigg\lfloor  \frac{f_n\{T_1-\sum\limits_{k=2}^K{y_{kn}}t_0^s\}}{\sum\limits_{k=2}^K{y_{kn}}\kappa _n\delta _{kn}^{\mathrm{F}}}\bigg\rfloor, \notag \\
& \quad \quad \quad   \quad \quad \bigg\lfloor\frac{T_1 {r_{nn^{\prime}}}}{{\sum\limits_{k=2}^{K-1}{y_{kn}}y_{(k+1) n^{\prime}} \sum\limits_{i=1}^{I}x_{ik}\varphi_i}}\bigg\rfloor ,\bigg\lfloor\frac{T_1 {r_{nn^{\prime}}}}{{\sum\limits_{k=2}^{K-1}{y_{kn}}y_{(k+1)n^{\prime}} \sum\limits_{i=1}^{I}x_{ik}\phi_{(i+1)}}}\bigg\rfloor\bigg\},
  \notag
 \\& \min\limits_{n\in  \mathcal{N}_c} \bigg\{\bigg\lfloor \frac{M\cdot M_n}{\sum\limits_{i=1}^Ix_{i1}(
 {\widetilde \varphi} _i
+\widetilde{\phi}_i + {\widetilde \sigma}_i + \beta_i)} \bigg\rfloor,\bigg\lfloor \frac{MT_1 {r_{nn^{\prime}}}}{\sum\limits_{i=1}^{I}x_{i1}\varphi_i}\bigg\rfloor,  
\bigg\lfloor \frac{MT_1 {r_{nn^{\prime}}}}{\sum\limits_{i=1}^{I}x_{i1}\phi_{(i+1)}}\bigg\rfloor,
   \bigg\lfloor     \frac{Mf_n(T_1-t_0^c)}{\kappa _n\delta _{1n}^{\mathrm{F}}} \bigg\rfloor\bigg\}\bigg\}.
 \label{b_prime}
\end{align}
 \end{figure*}
   
(1) When $0 < b\leqslant \min\{b_{\mathrm{th}}^c,b_{\mathrm{th}}^s\}$, the optimization Problem $\mathcal{P}3$ can be reformulated by 
\begin{align}
	\mathcal{P}3^{\prime }:\	&\underset{b}{\min}\  T_f(b)+\xi (b) \cdot T_1 \notag \\
		\mathrm{s.t.}\quad		
         & \mathrm{C3^{\prime }:}\   b\in \{1,2,...,B\}, \notag \\
         &\mathrm{C7}^{\prime}:	\   0 \leqslant   b_m\sum\limits_{i=1}^Ix_{i1}(
 {\widetilde \varphi} _i
+\widetilde{\phi}_i + {\widetilde \sigma}_i + \beta_i)\leqslant M_n,\notag\\
&\quad \quad \quad \quad \quad \quad \quad \quad \quad \quad \quad \quad  \quad \quad \quad \quad \quad  \forall n \in \mathcal{N}_c, \notag\\
&\mathrm{C8}^{\prime}:	\   0 \leqslant b\sum\limits_{i=1}^I (
 {\widetilde \varphi} _i
 +\widetilde{\phi}_i + {\widetilde \sigma}_i + \beta_i) \big\{x_{i1}
 + \sum\limits_{k=2}^K y_{kn}  \cdot \notag \\
          &\ \ \  \quad \quad  (x_{ik}-x_{i(k-1)}) \big\}\leqslant M_n,\ \forall n \in \mathcal{N}\setminus \mathcal{N}_c, \notag \\
 & \mathrm{C9}^{\prime }:\  b
  \leqslant \frac{f_n\{T_1-\sum\nolimits_{k=2}^K{y_{kn}}t_0^s\}}{\sum\limits_{k=2}^Ky_{kn} \kappa _n\delta _{kn}^{\mathrm{F}} },\  \forall   n\in\mathcal{N}\setminus \mathcal{N}_c, \notag  \\
   & \mathrm{C10}^{\prime }: \   b_m
 \leqslant \frac{f_n(T_1-t_0^c)}{\kappa _n\delta _{1n}^{\mathrm{F}}},\  \forall  n\in \mathcal{N}_c, \notag  \\
  &  \mathrm{C11^{\prime }:}\   b{\sum_{k=2}^{K-1}{y_{kn}}y_{(k+1) n^{\prime}} \sum_{i=1}^{I}x_{ik}\varphi_i}/{r_{nn^{\prime}}}\leqslant T_1, \notag \\
  &\ \ \ \ \ \ \ \ \ \forall    n,\ n^\prime\in\mathcal{N}\setminus \mathcal{N}_c, \notag  \\
                     &  \mathrm{C12^{\prime }:}\  b_m \sum_{i=1}^{I}x_{1k}\varphi_i/{r_{nn^{\prime}}}\leqslant T_1, \notag \\
                              &\ \ \ \ \ \ \ \ \ \forall   n\in\mathcal{N}_c,\ n^\prime\in\mathcal{N}\setminus \mathcal{N}_c, \notag  \\
                              &  \mathrm{C15^{\prime }:}\   b{\sum_{k=2}^{K-1}{y_{kn}}y_{(k+1)n^{\prime}} \sum_{i=1}^{I}x_{ik}\phi_{(i+1)}}/{r_{nn^{\prime}}}\leqslant T_1, \notag 
                              \\
                               &\ \ \ \ \ \ \ \ \ \forall    n,\ n^\prime\in \mathcal{N}\setminus \mathcal{N}_c, \notag  \\
                              &    \mathrm{C16^{\prime }:}\    b_m{ \sum_{i=1}^{I}x_{i1}\phi_{(i+1)}}/{r_{nn^{\prime}}}\leqslant T_1, \notag \\
                              &\ \ \ \ \ \ \ \ \ \forall   n\in\mathcal{N}_c,\ n^\prime\in\mathcal{N}\setminus \mathcal{N}_c.       
	\end{align}
 Herein, we try to obtain the optimal solution to Problem $\mathcal{P}3^{\prime }$. Since we fix the decision variables $x_{ik}$, $y_{kn}$, and $T_1$, the objective function $\mathcal{Z}_1$ is only related to the variable $b$, which is expressed in Eq. (\ref{zb1}).

 According to the shape of the objective function, we can find the optimal point $b_1$ in the feasible region (if it exists), which is expressed in Eq. (\ref{b_1}).
Moreover, according to Constraints C3$^\prime$ and C7$^\prime$\text{ - }C12$^\prime$, the boundary variable $b_v$ is defined in Eq. (\ref{b_prime}).
To further minimize the value of the objection function while remaining in the feasible domain, the optimal micro-batch size is

\begin{equation}
     b_1^*=\left\{ \begin{array}{l}
	1, \ b_1\leqslant 1,\\
	b_1,\ 1<b_1\leqslant \min \left\{ {b_{v}} ,B\right\}, \\
	\min \left\{ {b_{v}} ,B\right\},\ \min \left\{ {b_{v}},B \right\}<b_1,\\
\end{array} \right. 
\end{equation}
where $b_1 = \underset{b \in \{\lfloor \widetilde{b}_1 \rfloor, \lceil \widetilde{b}_1 \rceil\}}{\arg\min} T_f(b)+\xi (b) \cdot T_1$,
  \begin{align}
   & \widetilde{b}_1 = \bigg({B \cdot T_1}\big({{\sum\limits_{k=2}^K{\sum\limits_{n=1}^{N}{{y_{kn}\kappa_n\delta _{kn}^{\mathrm{F}}}/{f_n}}}}
+ \sum\limits_{n=1}^{N}\sum\limits_{n^{\prime}=1}^{N}\sum\limits_{i=1}^{I} \varLambda  } \notag\\
         & +\max_{n \in \mathcal{N}_c}\{\frac{1}{M}\kappa _n\delta _{1n}^{\mathrm{F}}
 /f_n+\sum\limits_{n^\prime=1}^Ny_{2n^\prime}\frac{1}{M} \cdot \sum\limits_{i=1}^{I}x_{i1}\varphi_i/r_{nn^\prime}\}
 +\notag \\
 &\max_{n \in  \mathcal{N}_c}\{\sum\limits_{n^\prime=1}^Ny_{2n^\prime}\frac{1}{M} \cdot \sum\limits_{i=1}^{I}x_{i1}\phi_{(i+1)}/r_{nn^\prime}\}
 \big)^{-1}\bigg)^{\frac{1}{2}}.
          \label{b_1}
  \end{align}
and
\begin{equation}
  \begin{aligned}
\varLambda =\frac{\sum\limits_{k=2}^{K-1}{y_{kn}}y_{(k+1) n^{\prime}}  x_{ik}\varphi_i+\sum\limits_{k=2}^{K-1}{y_{kn}}y_{(k+1)n^{\prime}} x_{ik}\phi_{(i+1)}}{r_{nn^{\prime}}}.
            \end{aligned}
\end{equation}

 \begin{figure*}[ht]
 \vspace{-2em}
\begin{equation}
\begin{aligned}
    &\mathcal{Z}_2 =  b \bigg\{ {\sum_{k=2}^K{\sum_{n=1}^{N}{\frac{y_{kn}\kappa_n(\delta _{kn}^{\mathrm{F}}+\delta _{kn}^{\mathrm{B}})}{f_n}}}}
+\sum\limits_{n=1}^{N}\sum_{n^{\prime}=1}^{N}\sum_{i=1}^{I} {\frac{\sum\limits_{k=2}^{K-1}{y_{kn}}y_{(k+1) n^{\prime}}  x_{ik}\varphi_i+\sum\limits_{k=2}^{K-1}{y_{kn}}y_{(k+1)n^{\prime}}x_{ik}\phi_{(i+1)}}{W_{nn^{\prime}}\cdot \log \left( 1+{p_n{(d_{nn^{\prime}})}^{-\gamma}}/{N_0} \right)}}
         \bigg\} +T_1 \bigg\lceil \frac{B-b}{b}\bigg \rceil -\\
          &     {{{(K-1)\bigg\{\frac{\kappa_n\delta _{kn}^{\mathrm{B}}b_{\mathrm{th}}^s}{f_n^s}}-(t_0^s+t_1^s)\bigg\}}}+\max_{n \in\mathcal{N}_c}\{b_m\frac{\kappa _n\delta _{1n}^{\mathrm{F}}}
 {f_n}+t_0^c + \sum\limits_{n^\prime=1}^Ny_{2n^\prime}b_m \sum\limits_{i=1}^{I}\frac{x_{i1}\varphi_i}{r_{nn^\prime}}\}\\
          &+\max_{n \in \mathcal{N}_c}\{\big(b_m -b_{\mathrm{th}}^c \big)\frac{\kappa _n\delta _{1n}^{\mathrm{B}}}
 {f_n}+t_1^c + \sum\limits_{n^\prime=1}^Ny_{2n^\prime}b_m  \sum\limits_{i=1}^{I}\frac{x_{i1}\phi_{(i+1)}}{r_{nn^\prime}}\}
          .
          \label{zb2}
  \end{aligned}
\end{equation}
 \end{figure*}

 (2) When $0<\max\{b_{\mathrm{th}}^s,b_{\mathrm{th}}^c\}\leqslant b $, the optimization Problem $\mathcal{P}3$ can be reformulated by
	\begin{align}
	\mathcal{P}3^{\prime \prime }:\	&\underset{b}{\min}\  T_f(b)+\xi (b) \cdot T_1 \notag \\
		\mathrm{s.t.}
        \quad 
          & \mathrm{C3}^{\prime},\ \mathrm{C7}^{\prime}\text{ - }\mathrm{C12}^{\prime},  \ \mathrm{C15}^{\prime}\text{ - }\mathrm{C16}^{\prime},\notag \\
                             & \mathrm{C13}^{\prime  }:\  b\leqslant \frac{f_n\{T_1-\sum\limits_{k=2}^K{y_{kn}}t_1^s\}}{\sum\nolimits_{k=2}^K{y_{kn}}\kappa _n\delta _{kn}^{\mathrm{B}}}+b_{\mathrm{th}}^s,\  \forall  n\in\mathcal{N}\setminus \mathcal{N}_c, \notag  \\
                              & \mathrm{C14}^{ \prime }:\  b
 \leqslant \frac{Mf_n(T_1-t_1^c)}{\kappa _n\delta _{1n}^{\mathrm{B}}}+Mb_{\mathrm{th}}^c,\  \forall  n\in\mathcal{N}_c.   
	\end{align}

 Similarly, to minimize the value of the objection function $\mathcal{Z}_2$ in Eq. (\ref{zb2}), the following variable $b_2^*$ is set to: 
\begin{equation}
     b_2^*=\left\{ \begin{array}{l}
	1, \ b_2 \leqslant 1,\\
	b_2,\ 1<b_2\leqslant\min \left\{ {b_v}^\prime,B \right\}, \\
	\min \left\{ {b_v}^\prime,B \right\},\ \min \left\{ {b_v}^\prime,B \right\}<b_2,\\
\end{array} \right. 
\end{equation}
where $b_2 = \underset{b \in \{\lfloor \widetilde{b}_2 \rfloor, \lceil \widetilde{b}_2 \rceil\}}{\arg\min} T_f(b)+\xi (b) \cdot T_1$. Herein, $\widetilde{b}_2$  and ${b_{v}}^\prime$ can be denoted by
\begin{align}
   & \widetilde{b}_2 =\bigg\lceil \bigg({B \cdot T_1}\big({{\sum\limits_{k=2}^K{\sum\limits_{n=1}^{N}{\frac{y_{kn}\kappa_n(\delta _{kn}^{\mathrm{F}}+\delta _{kn}^{\mathrm{B}})}{f_n}}}}
+\sum\limits_{n=1}^{N}\sum\limits_{n^{\prime}=1}^{N} \sum\limits_{i=1}^{I}
{\varLambda }} \notag\\
         & +\max_{n \in \mathcal{N}_c}\{\frac{1}{M}\kappa _n\frac{\delta _{1n}^{\mathrm{F}}}
 {f_n} + \sum\limits_{n^\prime=1}^Ny_{2n^\prime}\frac{1}{M} \cdot \sum\limits_{i=1}^{I}\frac{x_{i1}\varphi_i}{r_{nn^\prime}}\}
 +\max_{n \in \mathcal{N}_c}\{\frac{1}{M}\notag
 \\
 &\kappa_n \delta _{1n}^{\mathrm{B}}
 /f_n + \sum\limits_{n^\prime=1}^Ny_{2n^\prime}\frac{1}{M} \cdot \sum\limits_{i=1}^{I}x_{i1}\phi_{(i+1)}/r_{nn^\prime}\}
 \big)^{-1}\bigg)^{\frac{1}{2}}\bigg\rceil.
          \label{b_2}
  \end{align}
\begin{equation}
\begin{aligned}
{b_{v}}^\prime=& \min \bigg\{b_v, \min\limits_{ n \in \mathcal{N}\setminus \mathcal{N}_c}\bigg\{\bigg\lfloor  \frac{f_n\{T_1-\sum\limits_{k=2}^K{y_{kn}}t_0^s\}}{\sum\nolimits_{k=2}^K{y_{kn}}\kappa _n\delta _{kn}^{\mathrm{B}}}+b_{\mathrm{th}}^s\bigg\rfloor\bigg\}, \\
&\quad \quad \quad \quad \quad \quad   \min\limits_{ n \in \mathcal{N}_c}\bigg\{\bigg\lfloor\frac{Mf_n(T_1-t_1^c)}{\kappa _n\delta _{1n}^{\mathrm{B}}}+Mb_{\mathrm{th}}^c\bigg\rfloor \bigg\}
\bigg\}.
 \end{aligned}
 \label{b^prime}
\end{equation}

(3) Similarly, when $0<b_{\mathrm{th}}^c<b < b_{\mathrm{th}}^s $, the optimization Problem $\mathcal{P}3$ can be reformulated by
	\begin{align}
	\mathcal{P}3^{\prime \prime \prime }:\	&\underset{b}{\min}\  T_f(b)+\xi (b) \cdot T_1 \notag \\
		\mathrm{s.t.}
        \quad 
          & \mathrm{C3}^{\prime},\ \mathrm{C7}^{\prime}\text{ - }\mathrm{C12}^{\prime},  \ \mathrm{C15}^{\prime}\text{ - }\mathrm{C16}^{\prime},\notag \\
                              & \mathrm{C14}^{\prime  }:\  b
 \leqslant \frac{Mf_n(T_1-t_1^c)}{\kappa _n\delta _{1n}^{\mathrm{B}}}+Mb_{\mathrm{th}}^c,\  \forall  n\in\mathcal{N}_c,   
	\end{align}

To minimize the value of the objection function, the following variable $b_3^*$ is set to
\begin{equation}
     b_3^*=\left\{ \begin{array}{l}
	1, \ b_3\leqslant 1,\\
	b_3,\ 1<b_3\leqslant\min \left\{ {b_{v}}^{\prime \prime},B \right\}, \\
	\min \left\{ {b_{v}}^{\prime \prime},B \right\},\ \min \left\{ {b_{v}}^{\prime \prime},B \right\}<b_3,\\
\end{array} \right. 
\end{equation}
where $b_3 = \underset{b \in \{\lfloor \widetilde{b}_3 \rfloor, \lceil \widetilde{b}_3 \rceil\}}{\arg\min} T_f(b)+\xi (b) \cdot T_1$, and $\widetilde{b}_3$   can be expressed by
\begin{equation}
\begin{aligned}
   & \widetilde{b}_3 =\bigg\lceil \bigg({B \cdot T_1}\big({{\sum\limits_{k=2}^K{\sum\limits_{n=1}^{N}{\frac{y_{kn}\kappa_n\delta _{kn}^{\mathrm{F}}}{f_n}}}}
+\sum\limits_{n=1}^{N}\sum\limits_{n^{\prime}=1}^{N} \sum\limits_{i=1}^{I}
\varLambda 
          } +\max_{n \in \mathcal{N}_c}\\
         & \{\frac{1}{M}\frac{\kappa _n\delta _{1n}^{\mathrm{F}}}
 {f_n} + \sum\limits_{n^\prime=1}^Ny_{2n^\prime}\frac{1}{M} \cdot \sum\limits_{i=1}^{I}x_{i1}\varphi_i/r_{nn^\prime}\}
 +\max_{n \in \mathcal{N}_c}\{\frac{1}{M}\\
 &\kappa _n\delta _{1n}^{\mathrm{B}}
 /f_n + \sum\limits_{n^\prime=1}^Ny_{2n^\prime}\frac{1}{M} \cdot \sum\limits_{i=1}^{I}x_{i1}\phi_{(i+1)}/r_{nn^\prime}\}
 \big)^{-1}\bigg)^{\frac{1}{2}}\bigg\rceil,
          \label{b_3}
  \end{aligned}
  \end{equation}
 and ${b_{v}}^{\prime \prime}$ is denoted by
\begin{equation}
\begin{aligned}
&{b_{v}}^{\prime\prime}= \min \bigg\{ 
b_v,  \min\limits_{ n \in \mathcal{N}_c}\bigg\{  \bigg\lfloor \frac{nf_n(T_1-t_1^c)}{\kappa _n\delta _{1n}^{\mathrm{B}}}+Mb_{\mathrm{th}}^c\bigg\rfloor
\bigg\}\bigg\}.
 \end{aligned}
 \label{b^prime^prime}
\end{equation}


(4) Similarly, when $0<  b_{\mathrm{th}}^s<  b < b_{\mathrm{th}}^c $, the optimization Problem $\mathcal{P}3$ can be reformulated by
	\begin{align}
	\mathcal{P}3^{\prime \prime \prime \prime}:\	&\underset{b}{\min}\  T_f(b)+\xi (b) \cdot T_1 \notag \\
		\mathrm{s.t.}
        \quad 
          & \mathrm{C3}^{\prime},\ \mathrm{C7}^{\prime}\text{ - }\mathrm{C12}^{\prime},  \ \mathrm{C15}^{\prime}\text{ - }\mathrm{C16}^{\prime},\notag \\
                             & \mathrm{C13}^{\prime}:\  b\leqslant \frac{f_n\{T_1-\sum\limits_{k=2}^K{y_{kn}}t_1^s\}}{\sum\limits_{k=2}^K{y_{kn}}\kappa _n\delta _{kn}^{\mathrm{B}}}+b_{\mathrm{th}}^s,\  \forall  n\in\mathcal{N}\setminus \mathcal{N}_c.
	\end{align}

To minimize the value of the objection function, the following variable $b_4^*$ is set to 
\begin{equation}
     b_4^*=\left\{ \begin{array}{l}
	1, \ b_4\leqslant 1,\\
	b_4,\ 1<b_4\leqslant\min \left\{ {b_{v}}^{\prime \prime \prime},B \right\}, \\
	\min \left\{ {b_{v}}^{\prime \prime \prime},B \right\},\ \min \left\{ {b_{v}}^{\prime \prime \prime},B \right\}<b_4,\\
\end{array} \right. 
\end{equation}
where $b_4 = \underset{b \in \{\lfloor \widetilde{b}_4 \rfloor, \lceil \widetilde{b}_4 \rceil\}}{\arg\min} T_f(b)+\xi (b) \cdot T_1 $  and $\widetilde{b}_4$ is defined by
\begin{equation}
    \begin{aligned}
   & \widetilde{b}_4 =\bigg\lceil \bigg({B \cdot T_1}\big({{\sum\limits_{k=2}^K{\sum\limits_{n=1}^{N}
   {\frac{y_{kn}\kappa_n(\delta _{kn}^{\mathrm{F}}+\delta _{kn}^{\mathrm{B}})}{f_n}}}}
+\sum\limits_{n=1}^{N}\sum\limits_{n^{\prime}=1}^{N} \sum\limits_{i=1}^{I}
\varLambda 
          } \\
         & +\max_{n \in \mathcal{N}_c}\{\frac{1}{M}\frac{\kappa _n\delta _{1n}^{\mathrm{F}}}
 {f_n} + \sum\limits_{n^\prime=1}^Ny_{2n^\prime}\frac{1}{M}  \sum\limits_{i=1}^{I}\frac{x_{i1}\varphi_i}{r_{nn^\prime}}\}
 +\max_{n \in \mathcal{N}_c}\{ \sum\limits_{n^\prime=1}^N\\
 &y_{2n^\prime}\frac{1}{M} \cdot \sum\limits_{i=1}^{I}x_{i1}\phi_{(i+1)}/r_{nn^\prime}\}
 \big)^{-1}\bigg)^{\frac{1}{2}}\bigg\rceil.
          \label{b_4}
  \end{aligned}
\end{equation}
and ${b_{v}}^{\prime \prime \prime}$  is denoted by
\begin{equation}
\begin{aligned}
&{b_{v}}^{\prime\prime\prime}= \min \bigg\{ 
b_v,\min\limits_{n\in \mathcal{N}\setminus \mathcal{N}_c}\bigg\{\bigg\lfloor  \frac{f_n\{T_1-\sum\limits_{k=2}^K{y_{kn}}t_0^s\}}{\sum\nolimits_{k=2}^K{y_{kn}}\kappa _n\delta _{kn}^{\mathrm{B}}}+b_{\mathrm{th}}^s\bigg\rfloor
\bigg\}\bigg\}.
 \end{aligned}
 \label{b^prime^prime^prime}
\end{equation}

  Thus, the proof is completed.

  \section{Reformulation and Linearization of MSP Problem}
 To compute a lower bound for the min-sum
objective in Problem $\mathcal{P}4$, we need to reformulate and linearize the original problem. Above all, we define $\boldsymbol{\mu} = \big\{\mu_{knik}|\ k \in [1, K-1], n \in \mathcal{N}\big\}$ and $\boldsymbol{\zeta} = \big\{\zeta  _{knik (k+1) n^\prime i^\prime (k+1)}|\ k \in [1, K-1],\ n,\ n^\prime \in \mathcal{N}\big\}$. Moreover, $\mu _{knik}\in[0,1]$ and $\zeta  _{knik (k+1) n^\prime i^\prime (k+1)}\in[0,1]$ are continuous variables  and expressed by   
          \begin{align}
      &\mu _{knik}=y_{kn}\sum_{i=1}^Ix_{ik}, 
      \ \mu _{kni(k-1)}=y_{kn}\sum_{i=1}^Ix_{i(k-1)}, \notag \\
      &\zeta  _{knik (k+1) n^\prime i^\prime (k+1)}=y_{kn}\sum_{i=1}^Ix_{ik} \cdot y_{(k+1) n^{\prime}}\sum_{i^\prime=1}^Ix_{i^\prime (k+1)}.  
  \end{align}

In this section, considering the piecewise nature of the function induced by $b_{\mathrm{th}}^s$ and $b_{\mathrm{th}}^c$ and analyze each case separately, we have following situations:
  
  (1) When $0 < b\leqslant \min\{b_{\mathrm{th}}^c,b_{\mathrm{th}}^s\}$, the reformulation of min-sum part in $\mathcal{P}4$ can be expressed by
  \begin{align}
    \mathcal{P}4^\prime: &\underset{\{ \boldsymbol{\mu}, \boldsymbol{\zeta} \}}{\mathrm{min}} \quad	T_f(\boldsymbol{\mu}, \boldsymbol{\zeta}) =  
\mathop b\max\limits_{n \in \mathcal{N}_c}\{\sum_{n^{\prime}=1}^N{}\frac{\zeta _{1ni12n^{\prime}i^{\prime}2}\phi _{(i+1)}}{Mr_{nn^{\prime}}}\}+
 \notag \\
&
A(\boldsymbol{\mu },\boldsymbol{\zeta })+b\{\sum_{k=2}^K{\sum_{n=1}^N{\kappa _{n}\left( \left( \mu _{knik}-\mu _{kni(k-1)} \right) w_i \right) /f_{n}}}\}
,\notag\\
\mathrm{s.t.}\quad
          & \mathrm{C5^\dagger:}\quad  \sum_{i=1}^{I^\prime }\mu _{knik} \leqslant \sum_{i=1}^{I^\prime }\mu _{kni(k-1)},
          \notag \\
          & \quad \quad \quad  \quad \quad  \quad  \quad \quad  \quad 
 \  \forall k \in [2,K-1],\ I^\prime \in [1, I], \notag \\
          &\mathrm{C6^\dagger:}	\quad  \sum_{n=1}^N \mu _{knik}=1,\ \forall k \in [1,K], \notag \\
  		&\mathrm{C7^\dagger:}	\quad   0 \leqslant   m_{1}(\mu _{1ni1},b)\leqslant M_n,\ \forall n \in \mathcal{N}_c, \notag\\
    &\mathrm{C8^\dagger:}	\quad   0 \leqslant  \sum_{k=2}^K  
b(\mu _{knik}-\mu _{kni(k-1)})\cdot\notag\\&\quad \quad \quad \quad \quad (\widetilde{\varphi }_i+\widetilde{\phi }_i+\widetilde{\sigma }_i+\beta _i)
\leqslant M_n,\ \forall n \in \mathcal{N}\setminus \mathcal{N}_c,\notag\\
                    &\mathrm{C17^\dagger:}\ \sum_{k=1}^{K-1} \zeta  _{knik(k+1) n^\prime i^\prime (k+1)}=\mu _{(k+1) n^\prime i^\prime (k+1)}, \notag \\&
                    \quad \quad \quad \quad \quad \quad \quad \quad \quad \quad \quad \quad \quad \quad \quad \forall n,\ n^\prime\in\mathcal{N}, \notag \\
                    &\mathrm{C18^\dagger:}\ \sum_{k=2}^K \zeta  _{knik(k+1) n^\prime i^\prime (k+1) } = \mu _{knik}, \ \forall n,\ n^\prime\in\mathcal{N}. 
  \end{align}

 Moreover, $A(\boldsymbol{\zeta})$ can be expressed by
  \begin{align}
&A(\boldsymbol{\mu },\boldsymbol{\zeta })=b\sum_{n=1}^N{\sum_{n^{\prime}=1}^N{\sum_{k=2}^{K-1}\frac{{\zeta _{knik(k+1)n^{\prime}i^{\prime}(k+1)}}\left( \varphi _i+\phi _{(i+1)} \right)}{r_{nn^{\prime}}}}}+
\notag\\
&\mathop {b} \max\limits_{n \in \mathcal{N}_c}\{\frac{\kappa _n\mu _{1ni1}w_i}{Mf_{n}}+\sum_{n^{\prime}=1}^N{}\frac{\zeta _{1ni12n^{\prime}i^{\prime}2}\varphi _i}{Mr_{nn^{\prime}}}\}+K(t_{0}^{S}+t_{1}^{S}).
    \end{align}

  (2) When $0<\max\{b_{\mathrm{th}}^s,b_{\mathrm{th}}^c\}\leqslant b $, the optimization Problem $\mathcal{P}4$ can be reformulated by
 \begin{align}
    \mathcal{P}4^{\prime\prime}: &\underset{\{ \boldsymbol{\mu}, \boldsymbol{\zeta} \}}{\mathrm{min}} \quad	T_f(\boldsymbol{\mu}, \boldsymbol{\zeta}) = 
b\{\sum_{k=2}^K\sum_{n=1}^N\kappa _{n}\big(( \mu _{knik}-\mu _{kni(k-1)} ) w_i
         \notag\\& +( \mu _{knik}-\mu _{kni(k-1)} ) \rho _i\big)/f_{n}\}
-(K-1)\{\frac{\kappa _{n}\delta _{kn}^{\mathrm{B}}b_{\mathrm{th}}^{s}}{f_{n}}\}+\notag\\&b\max_{n \in \mathcal{N}_c} \{(\frac{1}{M}-b_{\mathrm{th}}^{c})\frac{\kappa _n\mu _{1ni1}\rho _i}{f_n}+\sum_{n^{\prime}=1}^N{}\frac{\zeta _{1ni12n^{\prime}i^{\prime}2}\phi _{(i+1)}}{Mr_{nn^{\prime}}}\}
\notag\\& +A(\boldsymbol{\mu },\boldsymbol{\zeta })
 , \notag \\
\mathrm{s.t.}\quad
          &  \mathrm{C5^\dagger}\text{ - } \mathrm{C8^\dagger}, \ \mathrm{C17^\dagger}, \ \mathrm{C18^\dagger}.
  \end{align}

  (3) Similarly, when $0<b_{\mathrm{th}}^c<b < b_{\mathrm{th}}^s $, the optimization Problem $\mathcal{P}4$ can be reformulated by
 \begin{align}
    \mathcal{P}4^{\prime\prime\prime}: &\underset{\{ \boldsymbol{\mu}, \boldsymbol{\zeta} \}}{\mathrm{min}} \quad	T_f(\boldsymbol{\mu}, \boldsymbol{\zeta}) =
\mathop b\max\limits_{n \in  \mathcal{N}_c}\{\sum_{n^{\prime}=1}^N\frac{\zeta _{1ni12n^{\prime}i^{\prime}2}\phi _{(i+1)}}{Mr_{nn^{\prime}}^{c}}\}
 \notag\\&  
+A(\boldsymbol{\mu },\boldsymbol{\zeta })+b\{\sum_{k=2}^K\sum_{n=1}^N\kappa _{n}\big(( \mu _{knik}-\mu _{kni(k-1)}) w_i+ \notag\\& ( \mu _{knik}-\mu _{kni(k-1)}) \rho _i\big)/f_{n}\}
 , \notag \\
\mathrm{s.t.}\quad
          & \mathrm{C5^\dagger}\text{ - } \mathrm{C8^\dagger}, \ \mathrm{C17^\dagger},\ \mathrm{C18^\dagger}.
  \end{align}

  (4) Similarly, when $0<  b_{\mathrm{th}}^s<  b < b_{\mathrm{th}}^c $, the optimization Problem $\mathcal{P}4$ can be reformulated by
 \begin{align}
    \mathcal{P}4^{\prime\prime\prime\prime}: &\underset{\{ \boldsymbol{\mu}, \boldsymbol{\zeta} \}}{\mathrm{min}} \quad	T_f(\boldsymbol{\mu}, \boldsymbol{\zeta}) =   
A(\boldsymbol{\mu },\boldsymbol{\zeta })-(K-1)\{\frac{\kappa _{n}\delta _{kn}^{\mathrm{B}}b_{\mathrm{th}}^{s}}{f_{n}}\}
+ \notag \\&
+b\{\sum_{k=2}^K{\sum_{n=1}^N{\kappa _{n}\left( \mu _{knik}-\mu _{kni(k-1)} \right) w_i/f_{n}}}\}+\mathop b \max\limits_{n \in  \mathcal{N}_c}
\notag \\
& 
\{(\frac{1}{M}-b_{\mathrm{th}}^{c})\frac{\kappa _n\mu _{1ni1}\rho _i}{f_n}+\sum_{n^{\prime}=1}^N{}\frac{\zeta _{1ni12n^{\prime}i^{\prime}2}\phi _{(i+1)}}{Mr_{nn^{\prime}}}\}, \notag \\
\mathrm{s.t.}\quad
          & \mathrm{C5^\dagger}\text{ - } \mathrm{C8^\dagger}, \ \mathrm{C17^\dagger},\ \mathrm{C18^\dagger}.
  \end{align}

By reformulating and linearizing the original min-sum part, it proves that the optimized value of $T_f(\boldsymbol{\mu}, \boldsymbol{\zeta})$ function does not exceed the optimal value of the original min-sum problem, thus providing a lower bound  \cite{FRIEZE198389}.

\bibliographystyle{IEEEtran}
\bibliography{citationlist}

\end{CJK}
\end{document}